\documentclass[11pt]{article}
\usepackage{bm}
\usepackage{amsfonts,amsmath,amssymb}
\usepackage{graphicx}
\usepackage{xspace}
\usepackage{boxedminipage}
\usepackage{url}
\usepackage{float}
\usepackage{enumerate,paralist}
\usepackage{accents}
\usepackage{setspace}
\usepackage{hyperref}
\usepackage{mathrsfs}
\usepackage[cm]{fullpage}%
\usepackage{paralist}%
\usepackage{scalerel}%
\usepackage[amsmath,thmmarks]{ntheorem}%
\usepackage{graphicx}
\usepackage[labelfont=bf]{caption}
\usepackage{subcaption}
\usepackage{colortbl}
\usepackage{hhline}
\usepackage{tikz}
\usepackage{multirow}
\usetikzlibrary{decorations.pathreplacing,angles,quotes}
\usepackage[normalem]{ulem}
\usepackage{afterpage}
\theoremseparator{.}%

\usepackage[margin=1in]{geometry}

\usepackage{salgorithm}

\usepackage{mleftright}%

\usepackage{hyperref}%
\hypersetup{%
   breaklinks,%
   ocgcolorlinks,
   colorlinks=true,%
   linkcolor=[rgb]{0.45,0.0,0.0},%
   citecolor=[rgb]{0,0,0.45}
}

\usepackage{slantsc}

\numberwithin{figure}{section}
\numberwithin{table}{section}
\numberwithin{equation}{section}%

\usepackage{titlesec}%
\titlelabel{\thetitle. }%

\newcommand{\HLinkShort}[2]{\hyperref[#2]{#1\ref*{#2}}}
\newcommand{\HLink}[2]{\hyperref[#2]{#1~\ref*{#2}}}
\newcommand{\HLinkPage}[2]{\hyperref[#2]{#1~\ref*{#2}%
      $_\text{p\pageref{#2}}$}}
\newcommand{\HLinkPageOnly}[1]{\hyperref[#1]{Page~\refpage*{#1}%
      $_\text{p\pageref{#1}}$}}

\newcommand{\HLinkSuffix}[3]{\hyperref[#2]{#1\ref*{#2}{#3}}}
\newcommand{\HLinkPageSuffix}[3]{\hyperref[#2]{#1\ref*{#2}%
      #3$_\text{p\pageref{#2}}$}}

\providecommand{\eqlab}[1]{}%
\renewcommand{\eqlab}[1]{\label{equation:#1}}

\theoremstyle{plain}%
\newtheorem{theorem}{Theorem}[section]

\newtheorem{lemma}[theorem]{Lemma}

\newtheorem{corollary}[theorem]{Corollary}

\theoremstyle{plain}%
\theorembodyfont{\upshape}%
\newtheorem*{remark:unnumbered}[theorem]{Remark}%
\newtheorem{definition}[theorem]{Definition}

\theoremheaderfont{\em}%
\theorembodyfont{\upshape}%
\theoremstyle{nonumberplain}%
\theoremseparator{}%
\theoremsymbol{\myqedsymbol}%
\newtheorem{proof}{Proof:}%

\newcommand{\myqedsymbol}{$\square$}

\def\compactify{\itemsep=0pt \topsep=0pt \partopsep=0pt \parsep=0pt}
\let\latexusecounter=\usecounter

{\begin{itemize}%
\setlength{\itemsep}{0pt}%
\setlength{\topsep}{0pt}%
\setlength{\partopsep}{0pt}%
\setlength{\parskip}{0pt}}%
{\end{itemize}}

\newcommand{\eps}{\varepsilon}%

\def\floor#1{\lfloor {#1} \rfloor}
\def\ceil#1{\lceil {#1} \rceil}
\def\script#1{\mathcal{#1}}

\def\etal{\text{et al.}\xspace}

\def\sep{\;|\;}
\def\card#1{|#1|}
\def\set#1{\{#1\}}
\def\E{\mathbf{E}}

\newcommand{\sol}{\ensuremath{\mathtt{sol}}}%
\def\path{\mathrm{path}}
\newcommand{\probSetCover}{\prob{Set\,Cover}{}\xspace}

\newcommand{\probCoverVerification}{\prob{Cover\,Verification}{}\xspace}
\newcommand{\Yes}{\prob{Yes}{}\xspace}
\newcommand{\No}{\prob{No}{}\xspace}
\newcommand{\Binomial}{\prob{B}{}}

\newcommand{\pbrcx}[1]{\left[ {#1} \right]}%
\newcommand{\Prob}[1]{\mathop{\mathbf{Pr}}\!\pbrcx{#1}}

\newcommand{\pth}[1]{\mleft({#1}\mright)}%

\def\cost{\mathrm{cost}}

\def\smp{\mathrm{smp}}

\def\rem{\mathrm{rem}}

\def\pr{\mathbf{Pr}}
\def\mypar#1{\smallskip\noindent\textbf{#1}}
\def\EltOf{\textsc{EltOf}}
\def\SetOf{\textsc{SetOf}}
\def\whp{\text{w.h.p.}}

\def\DCI{\mathfrak{D}}
\def\CI{\mathfrak{I}}

\def\replace{\textsf{r}}
\def\cost{\textsf{cost}}
\def\tough{\textsf{tough}}

\newcommand{\tldO}{\widetilde{O}}%
\newcommand{\tldOmega}{\widetilde{\Omega}}%
\newcommand{\tldTheta}{\widetilde{\Theta}}%

\newcommand{\prob}[1]{\textup{\fontencoding{T1}\fontfamily{ppl}\fontseries{m
}\fontshape{n}\selectfont #1}\xspace}

\def\sC{\script{C}}

\def\sM{\script{M}}

\def\sT{\script{T}}

\newcommand{\GSample}{\mathsf{S}}%
\newcommand{\sS}{\mathcal{F}}
\newcommand{\sE}{\mathcal{U}}
\newcommand{\CoverA}{\mathcal{D}}

\def\NP{\ensuremath{\mathrm{\mathbf{NP}}}}

\newcommand{\Algorithm}[1]{{\textsc{#1}}}

\newcommand{\algRandSC}{\Algorithm{LargeSetCover}\xspace}
\newcommand{\algIterSC}{\Algorithm{IterSetCover}\xspace}
\newcommand{\algSubSC}{\Algorithm{SmallSetCover}\xspace}
\newcommand{\genModifiedInst}{\Algorithm{GenModifiedInst}}
\newcommand{\offlineSC}{\Algorithm{OfflineSetCover}\xspace}%

\newcommand{\range}{\mathsf{r}}%

\newcommand{\inst}{I} 
\newcommand{\binst}{I^*} 
\newcommand{\minst}{I'} 
\newcommand{\binstF}{\mathcal{I}}  
\newcommand{\minstD}{\mathcal{D}} 
\newcommand{\alg}{\mathcal{A}} 
\newcommand{\apxF}{\alpha} 
\newcommand{\eltq}{\mathsf{Elt}}
\newcommand{\setq}{\mathsf{Set}}
\newcommand{\swap}{\mathsf{swap}}
\newcommand{\cand}{\mathsf{Candidate}}
\newcommand{\ans}{\mathsf{ans}}

\newcommand{\red}{\mathsf{red}}
\newcommand{\white}{\mathsf{white}}

\def\rnd{\mathsf{rnd}}
\def\up{u}
\def\low{l}

\def\rare{\mathsf{rare}}

\newcommand{\Null}{\mathsf{null}}

\begin{document}
\setlength{\abovedisplayskip}{3pt}
\setlength{\belowdisplayskip}{3pt}

\author{%
	Piotr Indyk  \thanks{CSAIL, MIT, \{indyk, mahabadi, vakilian, anak\}@mit.edu}
	\and
   Sepideh Mahabadi \footnotemark[1]
   \and%
   Ronitt Rubinfeld \thanks{CSAIL, MIT and  TAU, ronitt@csail.mit.edu}
   \and
   Ali Vakilian \footnotemark[1]
   \and%
   Anak Yodpinyanee \footnotemark[1]
}

\title{Set Cover in Sub-linear Time\thanks{This work was supported by the NSF grants, including No. CCF-1650733, CCF-1733808, CCF-1420692, IIS-1741137, and the Simons Investigator award.}}
\date{}
\maketitle

\setcounter{page}{1}%
\pagenumbering{arabic}%

\begin{abstract}

We study the classic set cover problem from the perspective of sub-linear algorithms.
Given access to a collection of $m$ sets over $n$ elements in the query model, we show that sub-linear algorithms derived from existing techniques have almost tight query complexities.

On one hand, first we show an adaptation of the streaming algorithm presented in \cite{imv-ttbsscp-15} to the sub-linear query model, that returns an $\apxF$-approximate cover using $\tldO(m(n/k)^{1/(\apxF-1)} + nk)$ queries to the input,  where $k$ denotes the value of a minimum set cover. We then complement this upper bound by proving that for lower values of $k$, the required number of queries is $\tldOmega(m(n/k)^{1/(2\apxF)})$, even for estimating the optimal cover size. Moreover, we prove that even checking whether a given collection of sets covers all the elements would require $\Omega(nk)$ queries.
These two lower bounds provide strong evidence that the upper bound is almost tight for certain values of the parameter $k$.

On the other hand, we show that this bound is not optimal for larger values of the parameter $k$, as there exists a $(1+\eps)$-approximation algorithm with $\tldO(mn/k\eps^2)$ queries. We show that this bound is essentially tight for sufficiently small constant $\eps$, by establishing a lower bound of $\tldOmega(mn/k)$ query complexity. 

Our lower-bound results follow by carefully designing two distributions of instances that are hard to distinguish. In particular, our first lower bound involves a probabilistic construction of a certain set system with a minimum set cover of size $\apxF k$, with the key property that a small number of ``almost uniformly distributed'' modifications can reduce the minimum set cover size down to $k$. Thus, these modifications are not detectable unless a large number of queries are asked. We believe that our probabilistic construction technique might find applications to lower bounds for other combinatorial optimization problems. 

\end{abstract}
\section{Introduction}
\probSetCover is a classic combinatorial optimization problem, in which we are given a set (universe) of $n$ elements $\sE = \set{e_1, \cdots, e_n}$ and a collection of $m$ sets $\sS = \set{S_1,\cdots, S_m}$. The goal is to find a \emph{set cover}  of $\sE$, i.e., a collection of sets in $\sS$ whose union is $\sE$, of minimum size. \probSetCover is a well-studied problem with applications in operations
research~\cite{gw-ceaas-97}, information retrieval and data mining~\cite{sg-mcsma-09}, learning theory~\cite{kv-iclt-1994}, web host analysis~\cite{ckt-mcmr-10}, and
many others. 
Recently, this problem and other related coverage problems have gained a lot of attention in the context of massive data sets, e.g., streaming model \cite{sg-mcsma-09, er-sssc-14, dimv-sccsc-14, imv-ttbsscp-15, cw-igpcs-16, aky-tbspscscp-16, mv-bsamcp-17, a-tsatmpsscp-17, bem-aosacp-17, imruvy-fscsm-17} or map reduce model \cite{kmvv-fgams-13,mz-rccsdsm-15,bem-dcms-16}.

Although the problem of finding an optimal solution is \NP-complete, a natural greedy algorithm which iteratively picks the ``best'' remaining set (the set that covers the most number of uncovered elements) is widely used.
The algorithm finds a solution of size at most $k \ln n$ where $k$ is the optimum cover size, and can be implemented to run in time linear in the input size. However, the input size itself could be as large as $\Theta(mn)$, so for large data sets even reading the input might be infeasible.

This raises a natural question: {\em is it possible to solve minimum set cover in sub-linear time?}
This question was previously addressed in \cite{no-ctaali-08, yyi-ictaammoop-12}, who showed that one can design constant running-time algorithms by simulating the greedy algorithm, under the assumption that the sets are of constant size and each element occurs in a constant number of sets.  
However,  those constant-time algorithms have a few drawbacks: they only provide a mixed multiplicative/additive guarantee (the output cover size is guaranteed to be at most $k\cdot \ln n+ \epsilon n$), 
 the dependence of their running times on the maximum set size is exponential, and they only output the (approximate) minimum set cover size, not the cover itself. From a different perspective, \cite{Koufogiannakis2014}  (building on~\cite{grigoriadis1995sublinear}) showed that an $O(1)$-approximate solution to the {\em fractional} version of the problem can be found in $\tldO(mk^2+nk^2)$ time\footnote{The method can be further improved to $\tldO(m+nk)$ (N. Young, personal communication).}. Combining this algorithm with the randomized rounding yields an $O(\log n)$-approximate solution to \probSetCover with the same complexity.

In this paper we initiate a systematic study of the complexity of sub-linear time algorithms for set cover with multiplicative approximation guarantees. 
Our upper bounds complement the aforementioned result of~\cite{Koufogiannakis2014} by presenting algorithms  which are fast when $k$ is {\em large}, as well as algorithms that provide more accurate solutions (even with a constant-factor approximation guarantee) that use a sub-linear number of {\em queries}\footnote{Note that polynomial {\em time} algorithm with sub-logarithmic approximation algorithms are unlikely to exist.}.
Equally importantly, we establish nearly matching lower bounds, some of which even hold for estimating the optimal cover size.
Our algorithmic results and lower bounds are presented in Table~\ref{t:results}.

\mypar{Data access model.} 
As in the prior work~\cite{no-ctaali-08,yyi-ictaammoop-12} on \prob{Set Cover}, our algorithms and lower bounds assume that the input can be accessed via the adjacency-list oracle.\footnote{In the context of graph problems, this model is also known as the {\em incidence-list model},  and has been studied extensively, see e.g.,~\cite{chazelle2005approximating, goel2013perfect, bhattacharya2015space}.}
More precisely, the algorithm has access to the following two oracles:

\begin{compactenum}
\item{\textbf{$\EltOf$:}} Given a set $S_i$ and an index $j$, the oracle returns the $j^{\mathrm{th}}$ element of $S_i$. If $j > \card{S_i}$, $\bot$ is returned.
\item{\textbf{$\SetOf$:}} Given an element $e_i$ and an index $j$, the oracle returns the $j^{\mathrm{th}}$ set containing $e_i$. If $e_i$ appears in less than $j$ sets, $\bot$ is returned.
\end{compactenum} 

This is a natural model, providing a ``two-way'' connection between the sets and the elements. Furthermore, for some graph problems modeled by \probSetCover (such as \prob{Dominating Set} or \prob{Vertex Cover}), such oracles are essentially equivalent to the aforementioned incident-list model  studied in sub-linear graph algorithms.
We also note that the other popular access model employing the {\em membership oracle}, where we can query whether an element $e$ is contained in a set $S$, is not suitable for \probSetCover, as it can be easily seen that even checking whether a feasible cover exists requires $\Omega(mn)$ time.

\subsection{Overview of our results}
In this paper we present algorithms and lower bounds for the  \probSetCover problem. The results are summarized  in  Table~\ref{t:results}.   
The $\NP$-hardness of this problem (or even its  $o(\log n)$-approximate version \cite{f-tasc-98,rs-sbepl-97,ams-acskr-06,m-pgcnsc-12,ds-aapr-14}) precludes the existence of highly accurate algorithms with fast running times, while (as we show) it is still possible to design algorithms with sub-linear query complexities and low approximation factors. The lower bound proofs hold for the running time of any algorithm approximation set cover assuming the defined data access model. 

We present two algorithms with sub-linear number of queries. First, we show that the streaming algorithm presented in \cite{imv-ttbsscp-15} can be adapted so that it returns an $O(\alpha)$-approximate cover using $\tldO(m(n/k)^{1/(\apxF-1)} + nk)$ queries, which could be {\em quadratically} smaller than $mn$. 
Second, we present a simple algorithm which is tailored to the case when the value of $k$ is large. This algorithm computes an $O(\log n)$-approximate cover in $\tldO(mn/k)$ {\em time} (not just query complexity). Hence, by combining it with the algorithm of~\cite{Koufogiannakis2014}, we get an $O(\log n)$-approximation algorithm that runs in time $\tldO(m+n\sqrt{m})$.

We complement the first result by proving that for low values of $k$, the required number of queries is $\tldOmega(m(n/k)^{1/(2\apxF)})$ even for estimating the size of the optimal cover. This shows that the first algorithm is essentially optimal for the values of $k$ where the first term in the runtime bound dominates.
Moreover, we prove that even the \probCoverVerification problem, which is checking whether a given collection of $k$ sets covers all the elements, would require $\Omega(nk)$ queries. This provides strong evidence that the term $nk$ in the first algorithm is unavoidable.
Lastly, we complement the second algorithm, by showing a lower bound of  $\tldOmega(mn/k)$ if the approximation ratio is a small constant.

\begin{table*}[t]
\centering
\renewcommand{\arraystretch}{1.5} 
\makebox[\textwidth][c]{
\begin{tabular}{|c||c|c|c|c|}
\hline \prob{Problem}
 &  \prob{Approximation} & \prob{Constraints} & \prob{Query Complexity} & \prob{Section}\\
\hline\hline
\multirow{4}{*}{\probSetCover}  &  $\apxF\rho+\eps$ & $\alpha \geq 2$ & $\tldO({1\over \eps}({m(\frac{n}{k})^{\frac{1}{\apxF-1}} + nk}))$ & \ref{sec:upperbound-small} \\
\hhline{~----}
&  $\rho+\eps$ & - & $\tldO(\frac{mn}{k\eps^2})$ & \ref{sec:upperbound-large}\\
\hhline{~----}
& $\apxF$ & $k < (\frac{n}{\log m})^{\frac{1}{4\apxF+1}}$ & $\tldOmega(m(\frac{n}{k})^{1/(2\apxF)})$ & \ref{sec:general-lb} \\
\hhline{~----}
& $\apxF$ & \begin{tabular}[x]{@{}c@{}}$\alpha\leq 1.01$\\$k = O({\frac{n} {\log m}})$\end{tabular} & $\tldOmega({\frac{mn}{k}})$ &\ref{sec:ind-reduction}\\
\hhline{-----}
\begin{tabular}[x]{@{}c@{}}\prob{Cover}\\\prob{Verification}\end{tabular}& - & $k\leq n/2$ & $\Omega(nk)$ & \ref{sec:verf} \\
\hhline{-----}
\end{tabular}
}
\caption{A summary of our algorithms and lower bounds. We use the following notation: $k \ge 1$ denotes the size of the optimum cover; $\alpha \ge 1$ denotes a parameter that determines the trade-off between the approximation quality and query/time complexities; $\rho \ge 1$ denotes the approximation factor of a ``black box'' algorithm for set cover used as a subroutine; We assume that $\alpha \leq \log n$ and $m\geq n$. 
}
\label{t:results}
\end{table*}

\subsection{Related work}\label{ss:rw}
Sub-linear algorithms for \prob{Set Cover} under the oracle model have been previously studied as an estimation problem; the goal is only to approximate the size of the minimum set cover rather than constructing one. Nguyen and Onak \cite{no-ctaali-08} consider \probSetCover under the oracle model we employ in this paper, in a specific setting 
where both the maximum cardinality of sets in $\sS$, and the maximum number of occurrences of an element over all sets, are bounded by some constants $s$ and $t$; this allows algorithms whose time and query complexities are constant, $(2^{(st)^4}/\eps)^{O(2^s)}$, containing no dependency on $n$ or $m$.
They provide an algorithm for estimating the size of the minimum set cover when, unlike our work, allowing both $\ln s$ multiplicative and $\eps n$ additive errors.
Their result has been subsequently improved to $(st)^{O(s)}/\eps^2$ by Yoshida \etal~\cite{yyi-ictaammoop-12}.
Additionally, the results of Kuhn \etal~\cite{kmw-pbns-06} on general packing/covering LPs in the distributed $\mathcal{LOCAL}$ model,
together with the reduction method of Parnas and Ron~\cite{parnas2007approximating}, implies estimating set cover size to within a $O(\ln s)$-multiplicative factor (with $\eps n$ additive error), can be performed in $(st)^{O(\log s \log t)}/\eps^4$ time/query complexities.

\probSetCover can also be considered as a generalization of the \prob{Vertex Cover} problem. The estimation variant of \prob{Vertex Cover} under the adjacency-list oracle model has been studied in \cite{parnas2007approximating,marko2006distance,orrr-nostaamvcs-12,yyi-ictaammoop-12}.
\probSetCover has been also studied in the sub-linear {\em space} context, most notably for the streaming model of computation~ \cite{sg-mcsma-09, er-sssc-14, cw-igpcs-16, aky-tbspscscp-16, a-tsatmpsscp-17, bem-aosacp-17, imruvy-fscsm-17, dimv-sccsc-14, imv-ttbsscp-15}.  In this model, there are algorithms that compute approximate set covers with only multiplicative errors. Our algorithms use some of the ideas introduced in the last two papers~\cite{dimv-sccsc-14, imv-ttbsscp-15}.

\subsection{Overview of the Algorithms}
The algorithmic results  presented in Section \ref{sec:upperbound},  use the techniques introduced for the streaming
\probSetCover problem by \cite{dimv-sccsc-14, imv-ttbsscp-15} to get new results in the context of sub-linear time algorithms for this problem. 
Two components previously used for the set cover problem in the context of streaming are \prob{Set Sampling} and \prob{Element Sampling}. Assuming the size of the minimum set cover is $k$, \prob{Set Sampling} randomly samples $\tldO(k)$ sets and adds them to the maintained solution. This ensures that all the elements that are well represented in the input (i.e., appearing in at least $m/k$ sets) are covered by the sampled sets. 
On the other hand, the \prob{Element Sampling} technique samples roughly $\tldO(k/\delta)$ elements, and finds a set cover for the sampled elements. It can be shown that the cover for the sampled elements covers a $(1-\delta)$ fraction of the original elements.

Specifically, the first algorithm performs a constant number of iterations. Each iteration uses element sampling to compute a ``partial'' cover, removes the elements covered by the sets selected so far and recurses on the remaining elements.  However, making this process work in sub-linear time (as opposed to sub-linear space) requires new technical development. For example, the algorithm of \cite{imv-ttbsscp-15} relies on the ability to test membership for a set-element pair, which generally cannot be efficiently performed in our model.

The second algorithm performs only one round of set sampling, and then identifies the elements that are not covered by the sampled sets, {\em without} performing a full scan of those sets.
This is possible because with high probability only those elements that belong to few input sets are not covered by the sample sets. Therefore, we can efficiently enumerate all  pairs $(e_i,S_j)$, $e_i \in S_j$, for those elements $e_i$ that were not covered by the sampled sets.
We then run a black box algorithm only on the set system induced by those pairs. This approach lets us avoid the $nk$ term present in the query and runtime bounds for the first algorithm, which makes the second algorithm highly efficient for large values of $k$. 

\subsection {Overview of the Lower Bounds}

\mypar{The \probSetCover lower bound for smaller optimal value $k$.}
We establish our lower bound for the problem of \emph{estimating} the size of the minimum set cover, by constructing two distributions of set systems. All systems in the same distribution share the same optimal set cover size, but these sizes differ by a factor $\apxF$ between the two distributions; thus, the algorithm is required to determine from which distribution its input set system is drawn, in order to correctly estimate the optimal cover size. Our distributions are constructed by a novel use of the probabilistic method. Specifically, we first probabilistically construct a set system called {\em median instance} (see~Lemma~\ref{lem:median-existance}): this set system has the property that (a) its minimum set cover size is $\apxF k$ and (b) a small number of changes to the instance reduces the minimum set cover size to $k$. We set the first distribution to be always this median instance.
Then, we construct the second distribution by a random process that performs the changes (depicted in Figure~\ref{fig:modified-instance}) resulting in a {\em modified instance}. This process distributes the changes almost uniformly throughout the instance, which implies that the changes are unlikely to be detected unless the algorithm performs a large number of queries.
We believe that this construction might find applications to lower bounds for other combinatorial optimization problems.

\mypar{The \probSetCover lower bound for larger optimal value $k$.} Our lower bound for the problem of \emph{computing} an approximate set cover leverages the construction above. We create a combined set system consisting of multiple modified instances all chosen independently at random, allowing instances with much larger $k$. By the properties of the random process generating modified instances, we observe that most of these modified instances have different optimal set cover solution, and that distinguishing these instances from one another requires many queries. Thus, it is unlikely for the algorithm to be able to compute an optimal solution to a large fraction of these modified instances, and therefore it fails to achieve the desired approximation factor for the overall combined instance.

\mypar{The \probCoverVerification lower bound for a cover of size $k$.}
For \probCoverVerification, however, we instead give an explicit construction of the distributions.
We first create an underlying set structure such that initially, the candidate sets contain all but $k$ elements. Then we may swap in each uncovered element from a non-candidate set. Our set structure is systematically designed so that each swap only modifies a small fraction of the answers from all possible queries; hence, each swap is hard to detect without $\Omega(n)$ queries. The distribution of valid set covers is composed of instances obtained by swapping in every uncovered element, and that of non-covers is similarly obtained but leaving one element uncovered. 
\section{Preliminaries for the Lower Bounds}
First, we formally specify the representation of the set structures of input instances, which applies to both \probSetCover and \probCoverVerification.

Our lower bound proofs rely mainly on the construction of instances that are hard to distinguish by the algorithm. To this end, we define the $\swap$ operation that exchanges a pair of elements between two sets, and how this is implemented in the actual representation.
\begin{definition}[$\swap$ operation]\label{def:swap-operation}
Consider two sets $S$ and $S'$. A swap on $S$ and $S'$ is defined over two elements $e, e'$ such that $e\in S\setminus S'$ and $e' \in S' \setminus S$, where $S$ and $S'$ exchange $e$ and $e'$. Formally, after performing $\swap(e,e')$, $S = (S\cup\set{{e'}})\setminus \set{e}$ and $S' = (S'\cup\set{{e}})\setminus \set{e'}$. As for the representation via $\EltOf$ and $\SetOf$, each application of $\swap$ only modifies $2$ entries for each oracle. That is, if previously $e=\EltOf(S,i)$, $S=\SetOf(e,j)$,  $e'=\EltOf(S',i')$, and $S'=\SetOf(e',j')$, then their new values change as follows:  $e'=\EltOf(S,i)$, $S'=\SetOf(e,j)$,  $e=\EltOf(S',i')$, and $S=\SetOf(e',j')$. 
\end{definition}
In particular, we extensively use the property that the amount of changes to the oracle's answers incurred by each $\swap$ is minimal. We remark that when we perform multiple $\swap$s on multiple \emph{disjoint} set-element pairs, every swap modifies distinct entries and do not interfere with one another.

Lastly, we define the notion of query-answer history, which is a common tool for establishing lower bounds for sub-linear algorithms under query models.
\begin{definition}By \emph{query-answer history}, we denote the sequence of query-answer pairs $\langle (q_1, a_1),$ $(q_2, a_2),$ $\ldots,$ $(q_r,a_r)\rangle$ recording the communication between the algorithm and the oracles, where each new query $q_{i+1}$ may only depend on the query-answer pairs $(q_1, a_1), \ldots, (q_i, a_i)$. In our case, each $q_i$ represents either a \SetOf~query or an \EltOf~query made by the algorithm, and each $a_i$ is the oracle's answer to that respective query according to the set structure instance.
\end{definition}

\section{Lower Bounds for the Set Cover Problem}\label{sec:simple-lb}

In this section, we present lower bounds for \prob{Set Cover} both for small values of the optimal cover size $k$ (in Section \ref{sec:lb-small}), and for large values of $k$ (in Section \ref{sec:ind-reduction}).
For low values of $k$, we prove the following theorem whose proof is postponed to Appendix~\ref{sec:general-lb}.

\begin{theorem} \label{thm:lowerbound-general}
For $2\leq k\leq  (\frac{n}{16\apxF\log m})^{1 \over 4\apxF+1}$ and $1< \alpha \leq \log n$, any randomized algorithm that solves the \probSetCover problem with approximation factor $\apxF$ and success probability at least $2/3$ requires $\tldOmega(m(n/k)^{1\over 2\apxF})$ queries.
\end{theorem}

Instead, in Section \ref{sec:lb-small} we focus on the simple setting of this theorem which applies to approximation protocols for distinguishing between instances with minimum set cover sizes $2$ and $3$, and show a lower bound of $\tldOmega(mn)$ (which is tight up to a polylogarithmic factor) for approximation factor $3/2$. 
This simplification is for the purpose of both clarity and also for the fact that the result for this case is used in Section \ref{sec:ind-reduction} to establish our lower bound for large values of $k$.

\mypar{High level idea.}
Our approach for establishing the lower bound is as follows. First, we construct a \emph{median instance} $\binst$ for \probSetCover, whose minimum set cover size is $3$. We then apply a randomized procedure $\genModifiedInst$, which slightly modifies the median instance into a new instance containing a set cover of size $2$. Applying Yao's principle, the distribution of the input to the deterministic algorithm is either $\binst$ with probability $1/2$, or a modified instance generated thru $\genModifiedInst(\binst)$, which is denoted by $\minstD(\binst)$, again with probability $1/2$. Next, we consider the execution of the deterministic algorithm. We show that unless the algorithm asks at least $\tldOmega(mn)$ queries, the resulting query-answer history generated over $\binst$ would be the same as those generated over instances constituting a constant fraction of $\minstD(\binst)$, reducing the algorithm's success probability to below $2/3$. More specifically, we will establish the following theorem.

\begin{theorem} \label{thm:lowerbound-2}
Any algorithm that can distinguish whether the input instance is $\binst$ or belongs to $\minstD(\binst)$ with probability of success greater than $2/3$, requires $\Omega(mn /\log m)$ queries.
\end{theorem}
\begin{corollary} \label{col:lowerbound-cover-3/2}
For $1< \alpha < 3/2$, and $k \leq 3$, any randomized algorithm that approximates by a factor of $\alpha$, the size of the optimal cover for the \probSetCover problem with success probability at least $2/3$ requires $\tldOmega(mn)$ queries.
\end{corollary}

For simplicity, we assume that the algorithm has the knowledge of our construction (which may only strengthens our lower bounds); this includes $\binst$ and $\minstD(\binst)$, along with their representation via $\textsc{EltOf}$ and $\textsc{SetOf}$. The objective of the algorithm is simply to distinguish them. Since we are distinguishing a distribution of instances $\minstD(\binst)$ against a single instance $\binst$, we may individually upper bound the probability that each query-answer pair reveals the modified part of the instance, then apply the union bound directly. However, establishing such a bound requires a certain set of properties that we obtain through a careful design of $\binst$ and $\genModifiedInst$. We remark that our approach shows the hardness of {\em distinguishing} instances with with different cover sizes. That is, our lower bound on the query complexity also holds for the problem of approximating the size of the minimum set cover (without explicitly finding one).

\medskip
Lastly, in Section~\ref{sec:ind-reduction} we provide a construction utilizing Theorem~\ref{thm:lowerbound-2} to extend Corollary~\ref{col:lowerbound-cover-3/2}, establish the following theorem on lower bounds for larger minimum set cover sizes.
\begin{theorem}\label{thm:mn-k}
For any sufficiently small approximation factor $\alpha\leq 1.01$ and $k= O(m/\log n)$, any randomized algorithm that computes an $\alpha$-approximation to the \probSetCover problem with success probability at least $0.99$ requires $\tldOmega(mn/k)$ queries.
\end{theorem}

\subsection{The \probSetCover Lower Bound for Small Optimal Value $k$}\label{sec:lb-small}
\subsubsection{Construction of the Median Instance $\binst$}
Let $\sS$ be a collection of $m$ sets such that (independently for each set-element pair $(S, e)$) $S$ contains $e$ with probability $1-p_0$, where $p_0 = \sqrt{9\log m \over n}$ (note that since we assume $\log m \leq n/c$ for large enough $c$, we can assume that $p_0\leq 1/2$). Equivalently, we may consider the incidence matrix of this instance: each entry is either $0$ (indicating $e \notin S$) with probability $p_0$, or $1$ (indicating $e \in S$) otherwise. We write $\sS \sim \binstF(\sE,p_0)$ denoting the collection of sets obtained from this construction. 

\begin{definition}[Median instance]\label{def:median-instance}
An instance of \prob{Set Cover}, $\inst$, is a \emph{median instance} if it satisfies all the following properties.
\begin{compactenum}[(a)]
\item No two sets cover all the elements. (The size of its minimum set cover is at least $3$.)\label{item:mi-sc-size}
\item For any two sets the number of elements not covered by the union of these sets is at most $18\log m$. \label{item:mi-uncovered-elem-size}
\item The intersection of any two sets has size at least $n/8$. \label{item:mi-intersection-size}
\item For any pair of elements $e, e'$, the number of sets $S$ s.t. $e\in S$ but $e'\notin S$ is at least ${m\sqrt{9\log m} \over 4\sqrt{n}}$. \label{item:mi-candidate-size}
\item For any triple of sets $S,S_1$ and $S_2$, $\card{(S_1\cap S_2)\setminus S} \leq 6\sqrt{n\log m}$. \label{item:mi-triple}
\item For each element, the number of sets that do not contain that element is at most $6 m\sqrt{\log m \over n}$. \label{item:mi-degneg}
\end{compactenum}
\end{definition}

\begin{lemma}\label{lem:median-existance}
There exists a median instance $\binst$ satisfying all properties from Definition~\ref{def:median-instance}. In fact, with high probability, an instance drawn from the distribution in which $\Prob{e\in S}=1-p_0$ independently at random, satisfies the median properties.
\end{lemma}
The proof of the lemma follows from standard applications of concentration bounds.  Specifically, it follows from the union bound and Lemmas~\ref{lem:solution-size}--\ref{lem:element-deg}, appearing in Appendix~\ref{sec:missing-proofs}.

\subsubsection{Distribution $\minstD(\binst)$ of Modified Instances $\minst$ Derived from $\binst$}

Fix a median instance $\binst$. We now show that we may perform $O(\log m)$ $\swap$ operations on $\binst$ so that the size of the minimum set cover in the modified instance becomes $2$.
Moreover, its incidence matrix differs from that of $\binst$ in $O(\log m)$ entries. Consequently, the number of queries to $\EltOf$ and $\SetOf$ that induce different answers from those of $\binst$ is also at most $O(\log m)$.

We define $\minstD(\binst)$ as the distribution of instances $\minst$ generated from a median instance $\binst$ by $\genModifiedInst(\binst)$ given below in Figure~\ref{fig:modified-instance} as follows. Assume that $\binst = (\sE,\sS)$. We select two different sets $S_1, S_2$ from $\sS$ uniformly at random; we aim to turn these two sets into a set cover. To do so, we swap out some of the elements in $S_2$ and bring in the uncovered elements. For each uncovered element $e$, we pick an element $e'\in S_2$ that is also covered by $S_1$. Next, consider the candidate set that we may exchange its $e$ with $e' \in S_2$:

\begin{definition}[Candidate set]\label{def:candidate-set}
For any pair of elements $e,e'$, the \emph{candidate set} of $(e,e')$ are all sets that contain $e$ but not $e'$. The collection of candidate sets of $(e,e')$ is denoted by $\cand(e,e')$. Note that $\cand(e,e') \neq \cand(e',e)$ (in fact, these two collections are disjoint).   
\end{definition}

\begin{figure}[h]
    \begin{center}
        \begin{algorithmEnv}
            \underline{{\genModifiedInst}%
               $\pth{ \binst=\pth{\sE, \sS}}$}: $\Bigl.$ \qquad
            \+\\ %

            $\sM \leftarrow \emptyset$ \\ 
            {\bf pick} two \emph{different} sets $S_1, S_2$ from $\sS$ uniformly at random\\
            \Foreach $e\in \sE\setminus (S_1 \cup S_2)$ \Do\\
            \> {\bf pick} $e'\in (S_1\cap S_2)\setminus \sM$ uniformly at random \\
            \> $\sM \leftarrow \sM \cup \set{e'}$\\
            \> {\bf pick} a random set $S$ in $\cand(e,e')$\\
            \> $\swap(e,e')$ between $S,S_2$
        \end{algorithmEnv}%
    \end{center}
    \vspace{-0.15in}
    \caption{The procedure of constructing a modified instance of $\binst$.}
    \label{fig:modified-instance}
\end{figure}

We choose a random set $S$ from $\cand(e,e')$, and  swap $e \in S$ with $e' \in S_2$ so that $S_2$ now contains $e$. We repeatedly apply this process for all initially uncovered $e$ so that eventually $S_1$ and $S_2$ form a set cover. We show that the proposed algorithm, \genModifiedInst, can indeed be executed without getting stuck.

\begin{lemma} \label{lem:2/3-well-defined}
The procedure \emph{$\genModifiedInst$} is well-defined under the precondition that the input instance $\binst$ is a median instance.
\end{lemma}
\begin{proof}
To carry out the algorithm, we must ensure that the number of the initially uncovered elements is at most that of the elements covered by both $S_1$ and $S_2$. This follows from the properties of median instances (Definition~\ref{def:median-instance}): $\card{\sE\setminus(S_1\cup S_2)} \leq 18\log m$ by property (\ref{item:mi-uncovered-elem-size}), and that the size of the intersection of $S_1$ and $S_2$ is greater than $n/8$ by property (\ref{item:mi-intersection-size}). That is, in our construction there are sufficiently many possible choices for $e'$ to be matched and swapped with each uncovered element $e$. Moreover, by property (\ref{item:mi-candidate-size}) there are plenty of candidate sets $S$ for performing $\swap(e,e')$ with $S_2$.
\end{proof}

\subsubsection{Bounding the Probability of Modification}\label{sec:almost-uniform}

Let $\minstD(\binst)$ denote the distribution of instances generated by $\genModifiedInst(\binst)$. If an algorithm were to distinguish between $\binst$ or $\minst \sim \minstD(\binst)$, it must find some cell in the $\EltOf$ or $\SetOf$ tables that would have been modified by $\genModifiedInst$, to confirm that $\genModifiedInst$ is indeed executed; otherwise it would make wrong decisions half of the time. We will show an additional property of this distribution: none of the entries of $\EltOf$ and $\SetOf$ are significantly more likely to be modified during the execution of $\genModifiedInst$. Consequently, no algorithm may strategically detect the difference between $\binst$ or $\minst$ with the desired probability, unless the number of queries is asymptotically the reciprocal of the maximum probability of modification among any cells.

Define $P_{\eltq-\setq}:\sE\times\sS \rightarrow [0,1]$ as the probability that an element is swapped by a set. More precisely, for an element $e\in \sE$ and a set $S\in \sS$, if $e\notin S$ in the median instance $\binst$, then $P_{\eltq-\setq}(e,S) = 0$; otherwise, it is equal to the probability that $S$ swaps $e$. We note that these probabilities are taken over $I'\sim\minstD(\binst)$ where $\binst$ is a fixed median instance. That is, as per Figure~\ref{fig:modified-instance}, they correspond to the random choices of $S_1, S_2$, the random matching $\sM$ between $\sE\setminus(S_1\cup S_2)$ and $S_1 \cap S_2$, and their random choices of choosing each candidate set $S$.  We bound the values of $P_{\eltq-\setq}$ via the following lemma.

\begin{lemma}\label{lem:p_cell_uniform}
For any $e \in \sE$ and $S \in \sS$, $P_{\eltq-\setq}(e,S)\leq {4800 \log m \over mn}$ where the probability is taken over $\minst\sim\minstD(\binst)$.
\end{lemma}
\begin{proof}
Let $S_1, S_2$ denote the first two sets picked (uniformly at random) from $\sS$ to construct a modified instance of $\binst$. For each element $e$ and a set $S$ such that $e\in S$ in the basic instance $\binst$,
\begin{align*}
P_{\eltq-\setq}(e,S) &=  \Prob{S = S_2} 
		\cdot \Prob{e\in S_1\cap S_2}\\ 
		&\cdot 
		\Prob{e \text{ matches to } \sE\setminus(S_1\cup S_2) \sep e\in S_1\cap S_2} \\
&+ \Prob{S\notin\set{S_1,S_2}}\\ 
		&\cdot \Prob{e\in S\setminus (S_1\cup S_2) \sep e\in S}\\	
		&\cdot \Prob{S \text{ swaps } e \text{ with } S_2 \sep e\in S\setminus (S_1\cup S_2)}.
\end{align*} 
where all probabilities are taken over $\minst \sim \minstD(\binst)$.
Next we bound each of the above six terms. Since we choose the sets $S_1, S_2$ randomly, $\Prob{S=S_2} = 1/m$. We bound the second term by $1$.
For the third term, since we pick a matching uniformly at random among all possible (maximum) matchings between $\sE\setminus (S_1\cup S_2)$ and $S_1\cap S_2$, by symmetry, the probability that a certain element $e\in S_1\cap S_2$ is in the matching is (by properties (\ref{item:mi-uncovered-elem-size}) and (\ref{item:mi-intersection-size}) of median instances),
\begin{align*}
{\card{\sE\setminus(S_1\cup S_2)} \over \card{S_1\cap S_2}} \leq {18\log m \over n/8} = {144\log m \over n}.
\end{align*}  
We bound the fourth term by $1$. To compute the fifth term, let $d_e$ denote the number of sets in $\sS$ that do not contain $e$. By property (\ref{item:mi-degneg}) of median instances, the probability that $e\in S$ is in $S\setminus (S_1\cup S_2)$ given that $S\notin \{S_1,S_2\}$ is at most,
\begin{align*}
{d_e (d_e-1) \over (m-1)(m-2)} \leq {36m^2\cdot {\log m \over n} \over m^2/2} = {72\log m \over n}.
\end{align*} 
Finally for the last term, note that by symmetry, each pair of matched elements $ee'$ is picked by $\genModifiedInst$ equiprobably. Thus, for any $e\in S\setminus (S_1 \cup S_2)$, the probability that each element $e'\in S_1 \cap S_2$ is matched to $e$ is ${1 \over \card{S_1\cap S_2}}$. By properties (\ref{item:mi-intersection-size})--(\ref{item:mi-triple}) of median instances, the last term is at most
\begin{align*}
&\sum_{e'\in (S_1 \cap S_2)\setminus S} \Prob{ee'\in \sM}\cdot
{1 \over \card{\cand(e,e')}}\\ 
&= \card{(S_1\cap S_2)\setminus S}\cdot {1 \over \card{S_1\cap S_2}}\cdot {1 \over \cand(e,e')} \\
&\leq 6\sqrt{n\log m} \cdot {1\over n/8} \cdot {1 \over {m\sqrt{9\log m} \over 4\sqrt{n}}}
={64 \over m}.
\end{align*}
Therefore,
\begin{align*}
P_{\eltq-\setq}(e,S) &\leq {1 \over m} \cdot 1 \cdot {144\log m \over n} + 1 \cdot {72\log m \over n} \cdot {64 \over m} \leq {4800 \log m \over mn}.
\end{align*} 
\end{proof}

\subsubsection{Proof of Theorem \ref{thm:lowerbound-2}}\label{sec:sc-lowerbound-proof}

Now we consider a median instance $\binst$, and its corresponding family of modified sets $\minstD(\binst)$. To prove the promised lower bound for randomized protocols distinguishing $\binst$ and $\minst\sim\minstD(\binst)$, we apply Yao's principle and instead show that no deterministic algorithm $\alg$ may determine whether the input is $\binst$ or $\minst \sim \minstD(\binst)$ with success probability at least $2/3$ using $r=o({mn\over\log m})$ queries. 
Recall that if $\alg$'s query-answer history $\langle (q_1, a_1),\ldots,(q_r, a_r) \rangle$ when executed on $\minst$ is the same as that of $\binst$, then $\alg$ must unavoidably return a wrong decision for the probability mass corresponding to $\minst$. 
We bound the probability of this event as follows.

\begin{lemma}\label{lem:elemnt-query}
Let $Q$ be the set of queries made by $\alg$ on $\binst$. Let $\minst\sim\minstD(\binst)$ where $\binst$ is a given median instance. Then the probability that $\alg$ returns different outputs on $\binst$ and $\minst$ is at most ${4800 \log m \over mn} \card{Q}$. 
\end{lemma}
\begin{proof}
Let $\alg(\inst)$ denote the algorithm's output for input instance $I$ (whether the given instance is $\binst$ or drawn from $\minstD(\binst)$). For each query $q$, let $\ans_{\inst}(q)$ denote the answer of $\inst$ to query $q$. Observe that since $\alg$ is deterministic, if all of the oracle's answers to its previous queries are all the same, then it must make the same next query. Combining this fact with the union bound, we may lower  bound the probability that $\alg$ returns the same outputs on $\binst$ and $\minst\sim\minstD(\binst)$ as follows:
\begin{align*}
\Prob{\alg(\binst) \neq \alg(\minst)} \leq  \sum_{t=1}^{\card{Q}} \Prob{\ans_{\binst}(q_t) \neq \ans_{\minst}(q_t)}.
\end{align*}
For each $q\in Q$, let $S(q)$ and $e(q)$ denote respectively the set and element queried by $q$. Applying  Lemma~\ref{lem:p_cell_uniform}, we obtain
\begin{align*}
\Prob{\alg(\binst) \neq \alg(\minst)} \leq  \sum_{t=1}^{\card{Q}} \Prob{\ans_{\binst}(q_t) \neq \ans_{\minst}(q_t)} \leq \sum_{t=1}^{\card{Q}} P_{\eltq-\setq}(e(q_t), S(q_t)) \leq {4800 \log m \over mn} \card{Q}.
\end{align*}
\end{proof}

\noindent
{\em Proof of Theorem~\ref{thm:lowerbound-2}.}
If $\alg$ does not output correctly on $\binst$, the probability of success of $\alg$ is less than $1/2$; thus, we can assume that $\alg$ returns the correct answer on $\binst$. This implies that $\alg$ returns an incorrect solution on the fraction of $\minst \sim \minst(\binst)$ for which $\alg(\binst) = \alg(\minst)$. Now recall that the distribution in which we apply Yao's principle consists of $\binst$ with probability $1/2$, and drawn uniformly at random from $\minstD(\binst)$ also with probability $1/2$. Then over this distribution, by Lemma~\ref{lem:elemnt-query},
\begin{align*}
\Prob{\alg \text{ succeeds}} \leq 1-{1\over 2}\pr_{\minst\sim\minstD(\binst)}[\alg(\binst) = \alg(\minst)] &\leq 1- {1 \over 2} \left(1- {4800 \log m \over mn} \card{Q} \right)\\ 
	&= {1\over 2} + {2400 \log m \over mn} \card{Q}.
\end{align*}
Thus, if the number of queries made by $\alg$ is less than ${mn \over 14400 \log m}$, then the probability that $\alg$ returns the correct answer over the input distribution is less than $2/3$ and the proof is complete.

\subsection{The \probSetCover Lower Bound for Large Optimal Value $k$.}\label{sec:ind-reduction}
Our construction of the median instance $\binst$ and its associated distribution $\minstD(\binst)$ of modified instances also leads to the lower bound of $\tldOmega(\frac{mn}{k})$ for the problem of computing an approximate solution to \probSetCover. This lower bound matches the performance of our algorithm for large optimal value $k$ and shows that it is tight for some range of value $k$, albeit it only applies to sufficiently small approximation factor $\alpha \leq 1.01$.

\mypar{Proof overview.}
We construct a distribution over \emph{compounds}: a compound is a \probSetCover instance that consists of $t=\Theta(k)$ smaller instances $\inst_1, \ldots, \inst_t$, where each of these $t$ instances is either the median instance $\binst$ or a random modified instance drawn from $\minstD(\binst)$. By our construction, a large majority of our distribution is composed of compounds that contains at least $0.2t$ modified instances $\inst_i$ such that, any deterministic algorithm $\alg$ must fail to distinguish $\inst_i$ from $\binst$ when it is only allowed to make a small number of queries. A deterministic $\alg$ can safely cover these modified instances with three sets, incurring a \cost~(sub-optimality) of $0.2t$. Still, $\alg$ may choose to cover such an $\inst_i$ with two sets to reduce its \cost, but it then must err on a different compound where $\inst_i$ is replaced with $\binst$. We track down the trade-off between the amount of \cost~that $\alg$ saves on these compounds by covering these $\inst_i$'s with two sets, and the amount of error on other compounds its scheme incurs. $\alg$ is allowed a small probability $\delta$ to make errors, which we then use to upper-bound the expected $\cost$ that $\alg$ may save, and conclude that $\alg$ still incurs an expected $\cost$ of $0.1t$ overall. We apply Yao's principle (for algorithms with errors) to obtain that randomized algorithms also incur an expected $\cost$ of $0.05t$, on compounds with optimal solution size $k\in[2t,3t]$, yielding the impossibility result for computing solutions with approximation factor $\alpha = \frac{k+0.1t}{k} > 1.01$ when given insufficient queries.

\subsubsection{Overall Lower Bound Argument}

\mypar{Compounds.} Consider the median instance $\binst$ and its associated distribution $\minstD(\binst)$ of modified instances for \probSetCover with $n$ elements and $m$ sets, and let $t = \Theta(k)$ be a positive integer parameter. We define a \emph{compound} $\CI = \CI(\inst_1, \inst_2, \ldots, \inst_t)$ as a set structure instance consisting of $t$ median or modified instances $\inst_1, \inst_2, \ldots, \inst_t$, forming a set structure $(\sE^t,\sS^t)$ of $n'\triangleq nt$ elements and $m'\triangleq mt$ sets, in such a way that each instance $\inst_i$ occupies separate elements and sets. Since the optimal solution to each instance $\inst_i$ is $3$ if $\inst_i = \binst$, and $2$ if $\inst_i$ is any modified instance, the optimal solution for the compound is $2t$ plus the number of occurrences of the median instance; this optimal objective value is always $\Theta(k)$.

\mypar{Random distribution over compounds.} Employing Yao's principle, we construct a distribution $\DCI$ of compounds $\CI(\inst_1, \inst_2, \ldots, \inst_t)$: it will be applied against any deterministic algorithm $\alg$ for computing an approximate minimum set cover, which is allowed to err on at most a $\delta$-fraction of the compounds from the distribution (for some small constant $\delta > 0$). For each $i\in[t]$, we pick $\inst_i = \binst$ with probability $c/{m \choose 2}$ where $c > 2$ is a sufficiently large constant. 
Otherwise, simply draw a random modified instance $\inst_i \sim \minstD(\binst)$. We aim to show that, in expectation over $\DCI$, $\alg$ must output a solution that of size $\Theta(t)$ more than the optimal set cover size of the given instance $\CI \sim \DCI$.

\mypar{$\alg$ frequently leaves many modified instances undetected.}
Consider an instance $\CI$ containing at least $0.95t$ modified instances. These instances constitute at least a $0.99$-fraction of $\DCI$: the expected number of occurrences of the median instance in each compound is only $c/{m \choose 2} \cdot t = O(t/m^2)$, so by Markov's inequality, the probablity that there are more than $0.05t$ median instances is at most $O(1/m^2) < 0.01$ for large $m$.
We make use of the following useful lemma, whose proof is deferred to Section~\ref{sec:bad-events}. In what follow, we say that the algorithm ``distinguishes'' or ``detects the difference'' between $\inst_i$ and $\binst$ if it makes a query that induces different answers, and thus may deduce that one of $\inst_i$ or $\binst$ cannot be the input instance. In particular, if $\inst_i = \binst$ then detecting the difference between them would be impossible.

\begin{lemma}\label{lem:undetected}
Fix $M \subseteq [t]$ and consider the distribution over compounds $\CI(\inst_1, \ldots, \inst_t)$ with $\inst_i \sim \minstD(\binst)$ for $i \in M$ and $\inst_i = \binst$ for $i \notin M$. If $\alg$ makes at most $o(\frac{mnt}{\log m})$ queries to $\CI$, then it may detect the differences between $\binst$ and at least $0.75t$ of the modified instances $\{\inst_i\}_{i\in M}$, with probability at most $0.01$.
\end{lemma}
We apply this lemma for any $|M| \geq 0.95t$ (although the statement holds for any $M$, even vacuously for $|M| < 0.75t$).
Thus, for $0.99\cdot 0.99 > 0.98$-fraction of $\DCI$, $\alg$ fails to identify, for at least $0.95t-0.75t=0.2t$ modified instances $\inst_i$ in $\CI$, whether it is a median instance or a modified instance. Observe that the query-answer history of $\alg$ on such $\CI$ would not change if we were to replace any combination of these $0.2t$ modified instances by copies of $\binst$. Consequently, if the algorithm were to correctly cover $\CI$ by using two sets for some of these $\inst_i$, it must unavoidably err (return a non-cover) on the compound where these $\inst_i$'s are replaced by copies of the median instance.

\mypar{Charging argument.}
We call a compound $\CI$ \emph{tough} if $\alg$ does not err on $\CI$, and $\alg$ fails to detect at least $0.2t$ modified instances; denote by $\DCI^\tough$ the conditional distribution of $\DCI$ restricted to tough instances. For tough $\CI$, let $\cost(\CI)$ denote the number of modified instances $\inst_i$ that the algorithm decides to cover with three sets. That is, for each tough compound $\CI$, $\cost(\CI)$ measures how far the solution returned by $\alg$ is, from the optimal set cover size. Then, there are at least $0.2t - \cost(\CI)$ modified instances $\inst_i$ that $\alg$ chooses to cover with only two sets despite not being able to verify whether $\inst_i = \binst$ or not. Let $R_\CI$ denote the set of the indices of these modified instances, so $|R_\CI| = 0.2t - \cost(\CI)$. By doing so, $\alg$ then errs on the \emph{replaced compound} $r(\CI, R_\CI)$, denoting the compound similar to $\CI$, except that each modified instance $\inst_i$ for $i \in R_\CI$ is replaced by $\binst$. In this event, we say that the tough compound $\CI$ \emph{charges} the replaced compound $r(\CI,R_\CI)$ via $R_\CI$. Recall that the total error of $\alg$ is $\delta$: this quantity upper-bounds the total probability masses of charged instances, which we will then manipulate to obtain a lower bound on $\E_{\CI\sim\DCI}[\cost(\CI)]$.

\mypar{Instances must share optimal solutions for $R$ to charge the same replaced instance.}
 Observe that many tough instances may charge to the same replaced instance: we must handle these duplicities. First, consider two tough instances $\CI^1\neq\CI^2$ charing the same $\CI_\replace = r(\CI^1,R) = r(\CI^2,R)$ via the same $R=R_{\CI^1} = R_{\CI^2}$. As $\CI^1\neq\CI^2$ but $r(\CI^1,R) = r(\CI^2,R)$, these tough instances differ on some modified instances with indices in $R$. Nonetheless, the query-answer histories of $\alg$ operating on $\CI^1$ and $\CI^2$ must be the same as their instances in $R$ are both indistinguishable from $\binst$ by the deterministic $\alg$. Since $\alg$ does not err on tough instances (by definition), both tough $\CI^1$ and $\CI^2$ must share the same optimal set cover on every instance in $R$. Consequently, for each fixed $R$, only tough instances that have the same optimal solution for modified instances in $R$ may charge the same replaced instance via $R$.

\mypar{Charged instance is much heavier than charging instances combined.}
By our construction of $\CI(\inst_1,\ldots,\inst_t)$ drawn from $\DCI$, $\Pr[\inst_i = \binst] = c/{m\choose 2}$ for the median instance. On the other hand, $\sum_{j=1}^\ell \Pr[\inst_i = \inst^j] \leq (1-c/{m\choose 2})\cdot(1/{m\choose 2}) < 1/{m\choose2}$ for modified instances $\inst^1, \ldots, \inst^\ell$ sharing the same optimal set cover, because they are all modified instances constructed to have the two sets chosen by \genModifiedInst~as their optimal set cover: each pair of sets is chosen uniformly with probability $1/{m\choose 2}$. Thus, the probability that $\binst$ is chosen is more than $c$ times the total probability that any $\inst^j$ is chosen. Generalizing this observation, we consider tough instances $\CI^1,\CI^2,\ldots,\CI^\ell$ charging the same $\CI_\replace$ via $R$, and bound the difference in probabilities that $\CI_\replace$ and any $\CI^j$ are drawn. For each index in $R$, it is more than $c$ times more likely for $\DCI$ to draw the median instance, rather than any modified instances of a fixed optimal solution. Then, for the replaced compound $\CI_\replace$ that $\alg$ errs, $p(\CI_\replace) \geq c^{|R|}\cdot \sum_{j=1}^\ell p(\CI^j)$ (where $p$ denotes the probability mass in $\DCI$, not in $\DCI^\tough$). In other words, the probability mass of the replaced instance charged via $R$ is always at least $c^{|R|}$ times the total probability mass of the charging tough instances.

\mypar{Bounding the expected $\cost$ using $\delta$.}
In our charging argument by tough instances above, we only bound the amount of charges on the replaced instances via a fixed $R$.
As there are up to $2^t$ choices for $R$, we scale down the total amount charged to a replaced instance by a factor of $2^t$, so that 
$\sum_{\tough~\CI} c^{|R_\CI|} p(\CI)/{2^t}$ lower bounds the total probability mass of the replaced instances that $\alg$ errs.

Let us first focus on the conditional distribution $\DCI^\tough$ restricted to tough instances. Recall that at least a $(0.98-\delta)$-fraction of the compounds in $\DCI$ are tough: $\alg$ fails to detect differences between $0.2t$ modified instances from the median instance with probability $0.98$, and among these compounds, $\alg$ may err on at most a $\delta$-fraction. So in the conditional distribution $\DCI^\tough$ over tough instances, the individual probability mass is scaled-up to $p^\tough(\CI) \leq \frac{p(\CI)}{0.98-\delta}$. Thus, 
\begin{align*}
\frac{\sum_{\tough~\CI} c^{|R_\CI|} p(\CI)}{2^t} &\geq \frac{\sum_{\tough~\CI} c^{|R_\CI|} (0.98-\delta) p^\tough(\CI)}{2^t} = \frac{(0.98-\delta)\E_{\CI\sim\DCI^\tough} \left[c^{|R_\CI|}\right]}{2^t}.
\end{align*}

As the probability mass above cannot exceed the total allowed error $\delta$, we have
\begin{align*}
\frac{\delta}{0.98-\delta}\cdot 2^t &\geq \E_{\CI\sim\DCI^\tough} \left[c^{|R_\CI|}\right] \geq \E_{\CI\sim\DCI^\tough} \left[c^{0.2t-\cost(\CI)}\right] \geq c^{0.2t-\E_{\CI\sim\DCI^\tough} [\cost(\CI)]},
\end{align*}
where Jensen's inequality is applied in the last step above. So,
\begin{align*}
\E_{\CI\sim\DCI^\tough} [\cost(\CI)] &\geq 0.2t-\frac{t+\log \frac{\delta}{0.98-\delta}}{\log c} = \left(0.2-\frac{1}{\log c}\right)t-\frac{\log \frac{\delta}{0.98-\delta}}{\log c} \geq 0.11t,
\end{align*}
for sufficiently large $c$ (and $m$) when choosing $\delta = 0.02$.

We now return to the expected $\cost$ over the entire distribution $\CI$. For simplicity, define $\cost(\CI)=0$ for any non-tough $\CI$. This yields $\E_{\CI\sim\DCI} [\cost(\CI)] \geq (0.98-\delta)\E_{\CI\sim\DCI^\tough} [\cost(\CI)] \geq (0.98-\delta)\cdot0.11t \geq 0.1t$, establishing the expected $\cost$ of any deterministic $\alg$ with probability of error at most $0.02$ over $\DCI$.

\mypar{Establishing the lower bound for randomized algorithms.}
Lastly, we apply Yao's principle\footnote{Here we use the Monte Carlo version where the algorithm may err, and use $\cost$ instead of the time complexity as our measure of performance. See, e.g., Proposition 2.6 in \cite{MR} and the description therein.} to obtain that, for any randomized algorithm with error probability $\delta/2 = 0.01$, its expected $\cost$ under the worst input is at least $\frac{1}{2}\cdot 0.1t = 0.05t$. Recall now that our $\cost$ here lower-bounds the sub-optimality of the computed set cover (that is, the algorithm uses at least $\cost$ more sets to cover the elements than the optimal solution does). Since our input instances have optimal solution $k\in[2t,3t]$ and the randomized algorithm returns a solution with $\cost$ at least $0.05t$ in expectation, it achieves an approximation factor of no better than $\alpha=\frac{k+0.05t}{k}>1.01$ with $o(\frac{mnt}{\log m})$ queries. Theorem~\ref{thm:mn-k} then follows, noting the substitution of our problem size: $\frac{mnt}{\log m} = \frac{(m'/t)(n'/t)t}{\log (m'/t)} = \Theta( \frac{m'n'}{k'\log m'})$.

\subsubsection{Proof of Lemma~\ref{lem:undetected}}\label{sec:bad-events}

First, we recall the following result from Lemma~\ref{lem:elemnt-query} for distinguishing between $\binst$ and a random $\minst\sim\minstD(\binst)$.

\begin{corollary}\label{lem:prob-found-opt}
Let $q$ be the number of queries made by $\alg$ on $\inst_i\sim\minstD(\binst)$ over $n$ elements and $m$ sets, where $\binst$ is a median instance. Then the probability that $\alg$ detects a difference between $\inst_i$ and $\binst$ in one of its queries is at most ${4800 q \log m \over mn}$. 
\end{corollary}

\mypar{Marbles and urns.} Fix a compound $\CI(\inst_1, \ldots, \inst_t)$. Let $s \triangleq {mn\over 4800 \log m}$, and then consider the following, entirely different, scenario. Suppose that we have $t$ urns, where each urn contains $s$ marbles. In the $i^\textrm{th}$ urn, in case $\inst_i$ is a modified instance, we put in this urn one $\red$ marble and $s-1$ $\white$ marbles; otherwise if $\inst_i = \binst$, we put in $s$ white marbles. Observe that the probability of obtaining a $\red$ marble by drawing $q$ marbles from a single urn \emph{without replacement} is exactly $q/s$ (for $q \leq s$). Now, we will relate the probability of drawing $\red$ marbles to the probability of successfully distinguishing instances. We emphasize that we are only comparing the probabilities of events for the sake of analysis, and we do not imply or suggest any direct analogy between the events themselves.

Corollary~\ref{lem:prob-found-opt} above bounds the probability that the algorithm successfully distinguishes a modified instance $\inst_i$ from $\binst$ with ${4800 q \log m \over mn} = q/s$. Then, the probability of distinguishing between $\inst_i$ and $\binst$ using $q$ queries, is bounded from above by the probability of obtaining a $\red$ marble after drawing $q$ marbles from an urn. Consequently, the probability that the algorithm distinguishes $3t/4$ 
instances is bounded from above by the probability of drawing the $\red$ marbles from at least $3t/4$ urns. Hence, to prove that the event of Lemma~\ref{lem:undetected} occurs with probability at most $0.01$, it is sufficient to upper-bound the probability that an algorithm obtains $3t/4$ $\red$ marbles by $0.01$.

Consider an instance of $t$ urns; for each urn $i \in [t]$ corresponding to a modified instance $\inst_i$, exactly one of its $s$ marbles is $\red$.
An algorithm may draw marbles from each urn, one by one without replacement, for potentially up to $s$ times. 
By the principle of deferred decisions, the $\red$ marble is equally likely to appear in any of these $s$ draws, independent of the events for other urns. 
Thus, we can create a tuple of $t$ random variables $\sT = (T_1,\ldots,T_t)$ such that for each $i \in [t]$, $T_i$ is chosen uniformly at random from $\{1, \ldots, s\}$. The variable $T_i$ represents the number of draws required to obtain the $\red$ marble in the $i^\textrm{th}$ urn; that is, only the $T_i^{\mathrm{th}}$ draw from the $i^\textrm{th}$ urn finds the $\red$ marble from that urn. In case $\inst_i$ is a median instance, we simply set $T_i = s+1$ indicating that the algorithm never detects any difference as $\inst_i$ and $\binst$ are the same instance.

We now show the following two lemmas in order to bound the number of $\red$ marbles the algorithm may encounter throughout its execution.
\begin{lemma}\label{lem:high-threshold-slabs}
Let $b > 3$ be a fixed constant and define $\sT_{\mathrm{high}} = \set{i ~|~ T_i \geq {s \over b}}$. If $t\geq 14b$, then $\card{\sT_{\mathrm{high}}} \geq (1-{2\over b})t$ with probability at least $0.99$.
\end{lemma}
\begin{proof}
Let $\sT_{\mathrm{low}}=\{1,\ldots,t\}\setminus\sT_{\mathrm{high}}$. Notice that for the $i^{\mathrm{th}}$ urn, $\Prob{i \in \sT_{\mathrm{low}}} < {1 \over b}$ independently of other urns, and thus $\card{\sT_{\mathrm{low}}}$ is stochastically dominated by \Binomial$(t,{1 \over b})$, the binomial distribution with $t$ trials and success probability ${1 \over b}$. Applying Chernoff bound, we obtain 
\begin{align*}
\Prob{\card{\sT_{\mathrm{low}}} \geq {2t\over b}} \leq e^{-{t \over 3b}} < 0.01.
\end{align*}
Hence, $\card{\sT_{\mathrm{high}}} \geq t- {2t \over b} = (1-{2\over b})t$ with probability at least $0.99$, as desired.
\end{proof}

\begin{lemma}\label{lem:sim-output}
If the total number of draws made by the algorithm is less than $(1-{3\over b}){st \over b}$, then with probability at least $0.99$, the algorithm will not obtain $\red$ marbles from at least ${t\over b}$ urns. 
\end{lemma}
\begin{proof} 
If the total number of such draws is less than $(1-{3\over b}){st \over b}$, then the number of draws from at least ${3t\over b}$ urns is less than ${s\over b}$ each. Assume the condition of Lemma~\ref{lem:high-threshold-slabs}: for at least $ (1-{2\over b})t$ urns, $T_i \geq {s\over b}$. That is, the algorithm will not encounter a $\red$ marble if it makes less than ${s\over b}$ draws from such an urn. Then, there are at least $t\over b$ urns with $T_i \geq {s\over b}$ from which the algorithm makes less than ${s\over b}$ draws, and thus does not obtain a $\red$ marble. Overall this event holds with probability at least $0.99$ due to Lemma~\ref{lem:high-threshold-slabs}.
\end{proof}

We substitute $b = 4$ and assume sufficiently large $t$. Suppose that the deterministic algorithm makes less than $(1-{3\over 4}){st \over 4} = \frac{st}{16}$ queries, then for a fraction of $0.99$ of all possible tuples $\sT$, there are $t/4$ instances $\inst_i$ that the algorithm fails to detect their differences from $\binst$: the probability of this event is lower-bounded by that of the event where the $\red$ marbles from those corresponding urns $i$ are not drawn. Therefore, the probability that the algorithm makes queries that detect differences between $\binst$ and more than $3t/4$ instances $\inst_i$'s is bounded by $0.01$, concluding our proof of Lemma~\ref{lem:undetected}.

\section{Sub-Linear Algorithms for the Set Cover Problem} \label{sec:upperbound}
In this paper, we present two different approximation algorithms for \probSetCover with sub-linear query in the oracle model: \algSubSC and \algRandSC.
Both of our algorithms rely on the techniques from the recent developments on \probSetCover in the streaming model. However, adopting those  techniques in the oracle model requires novel insights and technical development.

Throughout the description of our algorithms, we assume that we have access to a black box subroutine 
that given the full \probSetCover instance (where all members of all sets are revealed), returns a $\rho$-approximate 
solution\footnote{The approximation factor $\rho$ may take on any value between $1$ and $\Theta(\log n)$ 
depending on the computational model one assumes.}.

The first algorithm (\algSubSC) returns a $(\apxF \rho + \eps)$ approximate solution of the \prob{Set Cover} instance using $\tldO(\frac{1}{\eps} (m (\frac{n}{k})^{1\over \alpha-1} + nk))$ queries, while the second algorithm (\algRandSC) achieves an approximation factor of $(\rho+\eps)$ using $\tldO(\frac{mn}{k\eps^2})$ queries, where $k$ is the size of the minimum set cover. These algorithms can be combined so that the number of queries of the algorithm becomes asymptotically the minimum of the two:

\begin{theorem}
There exists a randomized algorithm for \probSetCover in the oracle model that \whp\footnote{An algorithm succeeds \emph{with high probability} ($\whp$) if its failure probability can be decreased to $n^{-c}$ for any constant $c > 0$ without affecting its asymptotic performance, where $n$ denotes the input size.}~computes an $O(\rho\log n)$-approximate solution and uses $\tldO( \min\{m\left({n\over k}\right)^{1/ \log n}+ nk \text{ , }\frac{mn}{k}\}) = \tldO(m + n \sqrt m)$ number of queries. 
\end{theorem}

\subsection{Preliminaries.}\label{sec:alg-prelim}
Our algorithms use the following two sampling techniques developed for \probSetCover in the streaming model~\cite{dimv-sccsc-14}: \prob{Element Sampling} and \prob{Set Sampling}. The first technique, \prob{Element Sampling}, states that in order to find a $(1-\delta)$-cover of $\sE$ \whp, it suffices to solve \probSetCover on a subset of elements of size $\tldO({\rho k \log m \over \delta})$ picked uniformly at random. It shows that we may restrict our attention to a subproblem with a much smaller number of elements, and our solution to the reduced instance will still cover a good fraction of the elements in the original instance. The next technique, \prob{Set Sampling}, shows that if we pick $\ell$ sets uniformly at random from $\sS$ in the solution, then each element that is not covered by any of picked sets $\whp$ only occurs in $\tldO({m\over \ell})$ sets in $\sS$; that is, we are left with a much sparser subproblem to solve. The formal statements of these sampling techniques are as follows. See~\cite{dimv-sccsc-14} for the proofs. 

\begin{lemma}[Element Sampling] \label{lem:element-sampling}
Consider an instance of \prob{Set Cover} on ($\sE$, $\sS$) whose optimal cover has size at most $k$.
Let $\sE_{\smp}$ be a subset of $\sE$ of size $\Theta\left({\rho k \log m \over \delta}\right)$ chosen uniformly at random, and let $\sC_{\smp} \subseteq \sS$ be a $\rho$-approximate cover for $\sE_{\smp}$. 
Then, \whp~$\sC_{\smp}$ covers at least $(1-\delta)\card{\sE}$ elements. 
\end{lemma}
\begin{lemma}[Set Sampling] \label{lem:set-sampling}
Consider an instance $(\sE,\sS)$ of \probSetCover. Let $\sS_{\rnd}$ be a collection of $\ell$ sets picked uniformly at random. Then, \whp~$\sS_{\rnd}$ covers all elements that appear in $\Omega({m \log n \over \ell})$ sets of $\sS$.
\end{lemma}

\subsection{First Algorithm: small values of $k$} \label{sec:upperbound-small}
The algorithm of this section is a modified variant of the streaming algorithm of \prob{Set Cover} in \cite{imv-ttbsscp-15} that works in the sub-linear query model. 
Similarly to the algorithm of~\cite{imv-ttbsscp-15}, our algorithm \algSubSC considers different guesses of the value of an optimal solution ($\eps^{-1}\log n$ guesses) and performs the core {\em iterative} algorithm \algIterSC for all of them in parallel. For each guess $\ell$ of the size of an optimal solution, the \algIterSC goes through $1/\alpha$ iterations and by applying \prob{Element Sampling}, guarantees that $\whp$ at the end of each iteration, the number of uncovered elements reduces by a factor of $n^{-1/\alpha}$. Hence, after $1/\alpha$ iterations all elements will be covered. Furthermore, since the number of sets picked in each iteration is at most $\ell$, the final solution has at most $\rho\ell$ sets where $\rho$ is the performance of the {\em offline} block \offlineSC{} that \algIterSC uses to solve the reduced instances constructed by \prob{Element Sampling}.

Although our general approach in \algIterSC is similar to the iterative core of the streaming algorithm of \probSetCover, there are challenges that we need to overcome so that it works {\em efficiently} in the query model.
Firstly, the approach of~\cite{imv-ttbsscp-15} relies on the ability to test membership for a set-element pair when executing its \emph{set filtering} subroutine: given a subset $\GSample$, the algorithm of~\cite{imv-ttbsscp-15} requires to compute $\card{S\cap \GSample}$ which cannot be implemented efficiently in the query model (in the worst case, requires $m|\GSample|$ queries).
Instead, here we employ the \emph{set sampling} which $\whp$ guarantees that the number of sets that contain an (yet uncovered) element is small. 

Next challenge is achieving $m(n/k)^{1/(\alpha-1)} +nk$ query bound for computing an $\alpha$-approximate solution. As mentioned earlier, both our approach and the algorithm of~\cite{imv-ttbsscp-15} need to run the algorithm in parallel for different guesses $\ell$ of the size of an optimal solution. 
However, since \algIterSC performs $m(n/\ell)^{1/(\alpha-1)} + n\ell$ queries, if \algSubSC invokes \algIterSC with guesses in an increasing order then the query complexity becomes $mn^{1/(\alpha-1)}+ nk$; on the other hand, if it invokes \algIterSC with guesses in a decreasing order then the query complexity becomes $m(n/k)^{1/(\alpha-1)}+mn$.  
To solve  this issue, \algSubSC performs in two stages: in the first stage, it finds a $(\log n)$-estimate of $k$ by invoking \algIterSC using $m+nk$ queries (assuming guesses are evaluated in an increasing order) and then in the second rounds it only invokes $\algIterSC$ with approximation factor $\alpha$ in the smaller $O(\log n)$-approximate region around the $(\log n)$-estimate of $k$ computed in the first stage. Thus, in our implementation, besides the desired approximation factor, $\algIterSC$ receives an upper bound and a lower bound on the size of an optimal solution.

Now, we provide a detailed description of $\algIterSC$.  It receives $\apxF, \epsilon, \low $ and $\up$ as its arguments, and it is guaranteed that 
the size of an optimal cover of the input instance, $k$, is in $[\low,\up]$.  
Note that the algorithm does not know the value of $k$ and the sampling techniques described in Section~\ref{sec:alg-prelim} rely on $k$. Therefore, the algorithm needs to find a $(1+\eps)$ estimate\footnote{The exact estimate that the algorithm works with is a $(1+ {\eps \over 2\rho\apxF})$ estimate.} of $k$ denoted as $\ell$. This can be done by trying all powers of $(1+\eps)$ in $[\low, \up]$. The parameter $\apxF$ denotes the trade-off between the query complexity and the approximation guarantee that the algorithm achieves. Moreover, we assume that the algorithm has access to a $\rho$-approximate black box solver of \probSetCover.  

\algIterSC first performs \prob{Set Sampling} to cover all elements that occur in $\tldOmega(m/\ell)$ sets.
Then it goes through $\apxF-2$ iterations and in each iteration, it performs \prob{Element Sampling} with parameter $\delta=\tldO(({\ell/n})^{1/(\apxF-1)})$. 
By Lemma~\ref{lem:element-sampling}, after $(\apxF-2)$ iterations, $\whp$ only $\ell\left({n\over \ell}\right)^{1/(\apxF-1)}$ elements remain uncovered, for which the algorithm finds a cover by invoking the {\em offline} set cover solver. The parameters are set so that all $(\apxF -1)$ instances that are required to be solved by the offline set cover solver (the $(\apxF-2)$ instances constructed by \prob{Element Sampling} and the final instance) are of size $\tldO(m\left({n\over \ell}\right)^{1/(\apxF-1)})$.

In the rest of this section, we show that \algSubSC $\whp$ returns an almost $(\rho\apxF)$-approximate solution of \prob{Set Cover($\sE, \sS$)} with query complexity $\tldO(m\left({n\over k}\right)^{1 \over \apxF-1}+nk)$ where $k$ is the size of a minimum set cover. 

\begin{theorem}\label{thm:sub-set-cover}
\emph{\algSubSC} outputs a $(\apxF\rho+\eps)$-approximate solution of $\probSetCover(\sE, \sS)$ using $\tldO(\frac{1}{\eps}(m(n/k)^{1\over \apxF-1}+ nk))$ number of queries \whp, where $k$ is the size of an optimal solution of $(\sE, \sS)$.
\end{theorem}

To analyze the performance of \algSubSC, first we need to analyze the procedures invoked by \algSubSC: \algIterSC and \offlineSC. The $\offlineSC$ procedure receives as an input a subset of elements $\GSample$ and an estimate on the size of an optimal cover of $\GSample$ using sets in $\sS$. The \offlineSC algorithm first determines all occurrences of $\GSample$ in $\sS$. Then it invokes a black box subroutine that returns a cover of size at most $\rho\ell$ (if there exists a cover of size $\ell$ for $\GSample$) for the reduced \probSetCover instance over $\GSample$. 

Moreover, we assume that all subroutines have access to the \EltOf{} and \SetOf{} oracles, $\card{\sE}$ and $\card{\sS}$. 

\begin{figure}[h]
    \begin{center}
        \begin{algorithmEnv}
            \underline{{\algIterSC}%
               $\pth{\apxF, \eps, \low, \up}$}: $\Bigl.$ \qquad
            \+\\ %
            \LineComment{Try all $(1+{\eps \over 2\apxF\rho})$-approximate
               guesses of $k$}\\
            \For $\ell \in \set{(1+{\eps \over 2\apxF\rho})^{i}\sep \log_{1+{\eps \over 2\apxF\rho}} \low \leq i \leq \log_{1+{\eps \over 2\apxF\rho}} \up}$\\ 
			\>\Do \textbf{in order}: %
            \\
            \> $\sol_\ell \leftarrow$ collection of $\ell$ sets picked uniformly at random\CodeComment{{Set Sampling}}\\%
            \>$\sE_{\rem} \leftarrow \sE \setminus \bigcup_{\range \in \sol_\ell} \range$ \CodeComment{$n\ell$ $\EltOf$}\\%
            \> \Repeat ($\alpha-2$) times \\%
            \>\> $\GSample \leftarrow$ sample of $\sE_{\rem}$ of size
            $\tldO(\rho \ell \left({n \over \ell}\right)^{1\over \apxF-1})$ %
            \\%
            \>\>
            $\CoverA \leftarrow \offlineSC{}(\GSample,\ell)$\\%
            \>\>\If $\CoverA = \Null$ \Then \\%
            \>\>\>\Break \CodeComment{Try the next value of $\ell$}\\%
            \>\> $\sol_\ell \leftarrow \sol_\ell \bigcup\CoverA$\\ 
            \>\>$\sE_{\rem} \leftarrow \sE_{\rem} \setminus \bigcup_{\range \in
               \CoverA}\range$
               \CodeComment{$\rho n\ell$ $\EltOf$}\\%
    		\>\If $\card{\sE_{\rem}} \leq \ell\left({n\over \ell}\right)^{1/(\alpha-1)}$ \CodeComment{Feasibility Test}\\%
            \>\>
            $\CoverA \leftarrow \offlineSC{}(\sE_{\rem},\ell)$\\%
            \>\>\If $\CoverA \neq \Null$ \Then \\%
            \>\>\>$\sol_\ell \leftarrow \sol_\ell \bigcup\CoverA$\\
            \>\>\> \Return $\sol_\ell$
        \end{algorithmEnv}%
    \end{center}
    \vspace{-0.15in}
    \caption{\algIterSC is the main procedure of the \algSubSC algorithm for the \probSetCover problem.}
    \label{fig:iter-sc-alg}
\end{figure}

\begin{figure}[h]
	\begin{center}
    	\begin{algorithmEnv}
    	        \underline{{\offlineSC}%
               $\pth{\GSample, \ell}$}: $\Bigl.$ \qquad
           		 \+\\ %
           	   $\sS_\GSample \leftarrow \emptyset$ \quad \\%
            	\For each element $e \in \GSample$ \Do \quad \\%
            	\> $\sS_{e}\leftarrow$ the collection of sets containing $e$\\%
            	
            	\> $\sS_\GSample \leftarrow \sS_\GSample \cup \sS_{e}$ \\%
            	
            	$\CoverA \leftarrow$ solution of size at most $\rho\ell$ for \prob{Set Cover} on ($\GSample, \sS_{\GSample}$) constructed by the black box solver \\%
    			\LineComment{If there exists no such cover, then $\CoverA = \Null$}\\%
    			\Return $\CoverA$
    	\end{algorithmEnv}
    \end{center}
    \vspace{-0.15in}
    \caption{$\offlineSC{(\GSample, \ell)}$ invokes a black box that returns a cover of size at most 
$\rho\ell$ (if there exists a cover of size $\ell$ for $\GSample$) for the \probSetCover instance that is the projection of $\sS$ over $\GSample$.}
    \label{fig:offline-sc-alg}
\end{figure}

\begin{lemma}\label{lem:offline-set-cover}
Suppose that each $e\in \GSample$ appears in $\tldO({m\over\ell})$ sets of $\sS$ and lets assume that there exists a set of $\ell$ sets in $\sS$ that covers $\GSample$. Then \emph{$\offlineSC(\GSample, \ell)$} returns a cover of size at most $\rho\ell$ of $\GSample$ using $\tldO({m\card{\GSample} \over \ell})$ queries. 
\end{lemma}
\begin{proof}
Since each element of $\GSample$ is contained by $\tldO({m\over \ell})$ sets in $\sS$, the information required to solve the reduced instance on $\GSample$ can be obtained by $\tldO({m\card{\GSample} \over \ell})$ queries (i.e. $\tldO({m\over \ell})$ \SetOf{} query per element in $\GSample$). 
\end{proof}

\begin{lemma}\label{lem:feasible-sol}
The cover constructed by the outer loop of \emph{$\algIterSC(\apxF,\eps, \low, \up)$} with the parameter $\ell>k$, $\sol_{\ell}$, $\whp$ covers $\sE$.
\end{lemma}
\begin{proof}
After picking $\ell$ sets uniformly at random, by \prob{Set Sampling} (Lemma~\ref{lem:set-sampling}), \whp~each element that is not covered by the sampled sets appears in $\tldO({m\over \ell})$ sets of $\sS$. 
Next, by \prob{Element Sampling} (Lemma~\ref{lem:element-sampling} with $\delta = \left({\ell\over n}\right)^{1/(\apxF-1)}$), at the end of each {\em inner} iteration, \whp~the number of uncovered elements decreases by a factor of $\left({\ell \over n}\right)^{1/(\apxF-1)}$. Thus after at most $(\apxF-2)$ iterations, $\whp$ less than $\ell\left({n\over \ell}\right)^{1/(\apxF-1)}$ elements remain uncovered. Finally, $\offlineSC$ is invoked on the remaining elements; hence, $\sol_{\ell}$ \whp~covers $\sE$. 
\end{proof}
\noindent

Next we analyze the query complexity and the approximation guarantee of \algIterSC. As we only apply \prob{Element Sampling} and \prob{Set Sampling} polynomially many times, all invocations of the corresponding lemmas during an execution of the algorithm must succeed \whp, so we assume their {\em high probability} guarantees for the proofs in rest of this section. 
\begin{lemma}\label{lem:iter-set-cover}
Given that $\low\leq k \leq {\up \over 1+\eps / (2\apxF\rho)}$, \whp~\emph{$\algIterSC(\apxF,\eps,\low,\up)$} finds a $(\rho\apxF+\eps)$-approximate solution of the input instance using $\tldO\left(\frac{1}{\eps}(m({n\over\low})^{1/(\apxF-1)} + nk)\right)$ queries.
\end{lemma}

\begin{proof}
Let $\ell_k = (1+{\eps \over 2\apxF\rho})^{\ceil{\log_{1+{\eps \over 2\apxF\rho}} k}}$ be the smallest power of $1+{\eps \over 2\apxF\rho}$ greater than or equal to $k$. Note that it is guaranteed that $\ell_k \in [\low, \up]$. By Lemma~\ref{lem:feasible-sol}, \algIterSC terminates with a guess value $\ell \leq \ell_k$. In the following we compute the query complexity of the run of \algIterSC with a parameter $\ell\leq \ell_k$.

\prob{Set Sampling} component picks $\ell$ sets and then update the set of elements that are not covered by those sets, $\sE_{\rem}$, using $O(n\ell)$ $\EltOf$ queries.
Next, in each iteration of the inner loop, the algorithm samples a subset $\GSample$ of size $\tldO\left({\ell ({n/\ell})^{1/(\apxF-1)}}\right)$ from $\sE_{\rem}$. Recall that, by \prob{Set Sampling} (Lemma~\ref{lem:set-sampling}), each $e\in \GSample \subset \sE_{\rem}$ appears in at most $\tldO({m/ \ell})$ sets. Since each element in $\sE_{\rem}$ appears in $\tldO(m/\ell)$, \offlineSC returns a cover $\CoverA$ of size at most $\rho\ell$ using $\tldO\left(m\left({n/\ell}\right)^{1/(\apxF-1)}\right)$ $\SetOf$ queries (Lemma~\ref{lem:offline-set-cover}). 
By the guarantee of \prob{Element Sampling} (Lemma~\ref{lem:element-sampling}), the number of elements in $\sE_{\rem}$ that are not covered by $\CoverA$ is at most $({\ell/n})^{1/(\apxF-1)}\card{\sE_{\rem}}$. Finally, at the end of each inner loop, the algorithm updates the set of uncovered elements $\sE_{\rem}$  by using $\tldO(n\ell )$ $\EltOf$ queries. The Feasibility Test which is passed $\whp$ for $\ell\leq \ell_k$ ensures that the final run of \offlineSC performs $\tldO(m(n/\ell)^{1/(\alpha-1)})$ $\SetOf$ queries. Hence, the total number of queries performed in each iteration of the outer loop of \algIterSC with parameter $\ell\leq \ell_k$ is $\tldO\left(m\left({n/ \ell}\right)^{1/(\apxF-1)} + n\ell\right)$.

By Lemma~\ref{lem:feasible-sol}, if $\ell_{k}\leq \up$, then the outer loop of \algIterSC is executed for $l \leq \ell \leq \ell_k$ before it terminates. Thus, the total number of queries made by \algIterSC is:
\begin{align*}
\sum_{i = \ceil{\log_{1+{\eps \over 2\apxF\rho}}\low}}^{\log_{1+{\eps \over 2\apxF\rho}}\ell_k} \tldO\left(m\left({n\over (1+{\eps \over 2\apxF\rho})^i}\right)^{1\over \apxF-1} + n(1+{\eps \over 2\apxF\rho})^i\right) 
&= \tldO\left(m\left({n\over \low}\right)^{1\over \apxF-1}\left(\log_{1+{\eps\over 2\apxF\rho}}{\ell_k \over \low}\right) + {n\ell_k \over \eps/(\rho\apxF)} \right)\\
&=\tldO\left({1\over \eps}\left(m\left({n\over \low}\right)^{1/(\apxF-1)} + nk\right)\right).
\end{align*}

Now, we show that the number of sets returned by \algIterSC is not more than $(\alpha\rho+\eps)\ell_k$. \prob{Set Sampling} picks $\ell$ sets and each run of $\offlineSC{}$ returns at most $\rho\ell$ sets. 
Thus the size of the solution returned by \algIterSC is at most $(1+(\alpha-1)\rho)\ell_k < (\alpha\rho+\eps)k$.
\end{proof}

Next, we prove the main theorem of the section.

\begin{figure}[h]
    \begin{center}
        \begin{algorithmEnv}
            \underline{{\algSubSC}%
               $\pth{\apxF, \eps}$}: $\Bigl.$ \qquad
            \+\\ %
            
            $\sol \leftarrow \algIterSC(\log n, 1, 1,n)$ \\
            $k' \leftarrow \card{\sol}$ \CodeComment{Find a $\rho \log n$ estimate of $k$.} \\
            \Return $\algIterSC(\apxF, \epsilon, \floor{k' \over \rho\log n},\ceil{k'(1+{\eps \over 2\apxF\rho})})$%
        \end{algorithmEnv}%
    \end{center}
    \vspace{-0.15in}
    \caption{The description of the \algSubSC algorithm.}
    \label{fig:sub-sc-alg}
\end{figure}

\noindent
{\em Proof of Theorem~\ref{thm:sub-set-cover}.}
The algorithm $\algSubSC$ first finds a $(\rho\log n)$-approximate solution of $\prob{Set Cover}(\sE,\sS)$, $\sol$, with $\tldO(m+nk)$ queries by calling $\algIterSC(\log n, 1, 1, n)$. Having that $k\leq k' = \card{\sol}\leq (\rho\log n) k$, the algorithm calls $\algIterSC$ with $\apxF$ as the approximation factor and $[\floor{k'/(\rho\log n)}, \ceil{k'(1+{\eps \over 2\apxF\rho})}]$ as the range containing $k$. By Lemma~\ref{lem:iter-set-cover}, the second call to \algIterSC in \algSubSC returns a $(\apxF\rho+\eps)$-approximate solution of $\prob{Set Cover}(\sE, \sS)$ using the following number of queries:
\begin{align*}
\tldO\left({ 1\over \eps}\left(m\left({n\over {k/(\rho\log n)}}\right)^{1\over\apxF-1} + nk\right)\right) = \tldO\left({1\over \eps}\left(m\left({n\over k}\right)^{1\over \apxF-1} + nk\right)\right).
\end{align*}


\subsection{Second Algorithm: large values of $k$} \label{sec:upperbound-large}
The second algorithm, \algRandSC, works strictly better than \algSubSC for large values of $k$ ($k\geq \sqrt{m}$).  
The advantage of \algRandSC is that it does not need to update the set of uncovered elements at any point and simply avoids the additive $nk$ term in the query complexity bound; the result of Section~\ref{sec:verf} suggests that the $nk$ term may be unavoidable if one wishes to maintain the uncovered elements.
Note that the guarantees of \algRandSC is that at the end of the algorithm, $\whp$ the ground set $\sE$ is covered.
  
The algorithm \algRandSC, given in Figure~\ref{fig:unweighted-alg}, first randomly picks $\eps\ell/3$ sets. By \prob{Set Sampling} (Lemma~\ref{lem:set-sampling}), \whp{} every element that occurs in $\tldOmega({m/(\eps\ell)})$ sets of $\sS$ will be covered by the picked sets. It then solves the \prob{Set Cover} instance over the elements that occur in $\tldO({m/(\eps\ell)})$ sets of $\sS$ by an offline solver of \probSetCover using $\tldO({m/(\eps\ell)})$ queries; note that this set of elements may include some already covered elements. In order to get the promised query complexity, \algRandSC enumerates the guesses $\ell$ of the size of an optimal set cover in the decreasing order. The algorithm returns feasible solutions for $\ell \geq k$ and once it cannot find a feasible solution for $\ell$, it returns the solution constructed for the previous guess of $k$, i.e., $\ell (1+{\eps/(3\rho)})$.
   
Since \algRandSC performs \prob{Set Sampling} for $\tldO(\eps^{-1})$ iterations, $\whp$~the total query complexity of \algRandSC is $\tldO({mn/(k\eps^2)})$. 
 
Note that testing whether the number of occurrences of an element is ${\tldO(m/(\eps\ell))}$ only requires a single query, namely $\SetOf(e,{cm\log n \over \eps\ell})$. 
\begin{figure}[h]
    \begin{center}
        \begin{algorithmEnv}
            \underline{\algRandSC%
               $\pth{\eps}$}: $\Bigl.$ \qquad
            \+\\
            \LineComment{Try all $(1+{\eps \over 3\rho})$-approximate guesses of $k$}\\
            \For $\ell \in \set{(1+{\eps \over 3\rho})^{i}\sep 0\leq i \leq \log_{1+{\eps \over 3\rho}} n}$\\ 
			\> \Do \textbf{in the decreasing order}:
            \\
            \> $\rnd_\ell \leftarrow$ collection of ${\eps\ell\over 3}$ sets picked uniformly at random\CodeComment{Set Sampling}\\%
            \> $\sS_\rare \leftarrow \emptyset$\CodeComment{intersection with {\em rare} elements}\\
            \> \For $e \in \sE$ \Do \quad \\%
            \>\> \If $e$ appears in $< cm\log n\over \eps\ell$ sets
            \Then \CodeComment{Size Test: $\SetOf(e,{cm\log n\over \eps\ell})$} \\%

			\>\>\> $\sS_{e}\leftarrow$ collection of sets containing $e$  \CodeComment{$\tldO({m \over \eps\ell})$ $\SetOf$ queries}\\%
            \>\>\> $\sS_\rare \leftarrow \sS_\rare \cup \sS_{e}$, \quad $\GSample \leftarrow \GSample \bigcup \set{e}$\\
            \> $\CoverA \leftarrow$ solution of \prob{Set Cover($\GSample, \sS_{\rare}$)} returned by a $\rho$-approximate black box algorithm  \\%
                        \>\If $\card{\CoverA} \leq \rho\ell$ \Then \\
                        \>\> $\sol \leftarrow \rnd_{\ell} \cup \CoverA$\\
            \>\Else \Return $\sol$ \CodeComment{solution for the previous value of $\ell$}
        \end{algorithmEnv}%
    \end{center}
    \vspace{-0.15in}
    \caption{A $(\rho+\eps)$-approximation algorithm for the \probSetCover problem. We assume that the algorithm has access to  \EltOf{} and \SetOf{} oracles for $\probSetCover(\sE, \sS)$, as well as $\card{\sE}$ and $\card{\sS}$.}
    \label{fig:unweighted-alg}
\end{figure}

We now prove the desired performance of \algRandSC.

\begin{lemma}\label{lem:approx-rand}
\emph{$\algRandSC$} returns a $(\rho+\eps)$-approximate solution of $\probSetCover(\sE, \sS)$ \whp
\end{lemma}
\begin{proof}
The algorithm $\algRandSC$ tries to construct set covers of decreasing sizes until it fails. Clearly, if $k\leq \ell$ then the black box algorithm finds a cover of size at most $\rho\ell$ for any subset of $\sE$, because $k$ sets are sufficient to cover $\sE$. In other words, the algorithm does not terminate with $\ell\geq k$. Moreover, since the algorithm terminates when $\ell$ is smaller than $k$, the size of the set cover found by \algRandSC is at most $({\eps\over 3}+\rho)(1+{\eps\over 3\rho})\ell < ({\eps\over 3}+\rho)(1+{\eps\over 3\rho})k < (\rho+\eps)k$.
\end{proof}

\begin{lemma}\label{lem:complexity-rand}
The number of queries made by \emph{$\algRandSC$} is $\tldO({mn\over k\eps^2})$.
\end{lemma}
\begin{proof}
The value of $\ell$ in any {\em successful} iteration of the algorithm is greater than ${k/(\rho+\eps)}$; otherwise, the size of the solution constructed by the algorithm is at most $(\rho+\eps)\ell < k$ which is a contradiction.

\prob{Set Sampling} guarantees that $\whp$ each uncovered element appears in $\tldTheta({m/\eps\ell})$ sets and thus the algorithm needs to perform $\tldO({mn\over \eps\ell})$ \SetOf{} queries to construct $\sS_{\rare}$.
Moreover, the number of required queries in the \emph{size test} step is $O(n)$ because we only need one \SetOf{} query per each element in $\sE$.
Thus, the query complexity of $\algRandSC(\eps)$ is bounded by
\begin{align*}
&\sum_{i = \log_{1+{\eps \over 3\rho}} {k\over \rho+\eps}}^{ \log_{1+{\eps \over 3\rho}}n} \tldO\left(n + {m n\over \eps(1+{\eps \over 3\rho})^i}\right)
= \tldO\left((n + {mn \over \eps k}) \log_{1+{\eps\over 3\rho}}{n \over k}\right) = \tldO\left({mn \over k\eps^2}\right).
\end{align*} 
\end{proof}
\section{Lower Bound for the Cover Verification Problem} \label{sec:verf}

In this section, we give a tight lower bound on a feasibility variant of the \probSetCover problem which we refer to as $\probCoverVerification$. In $\probCoverVerification(\sE,\sS,\sS_k)$, besides a collection of $m$ sets $\sS$ and $n$ elements $\sE$, we are given indices of $k$ sets $\sS_k \subseteq \sS$, and the goal is to determine whether they are covering the whole universe $\sE$ or not. We note that, throughout this section, the parameter $k$ is a candidate for, but not necessarily the value of, the size of the minimum set cover.

A naive approach for this decision problem is to query all elements in the given $k$ sets and then check whether they cover $\sE$ or not; this approach requires $O(nk)$ queries. However, in what follows we show that this approach is tight and no \emph{randomized} protocol can decide whether the given $k$ sets cover the whole universe with probability of success at least $0.9$ using $o(nk)$ queries.    
\begin{theorem}\label{thm:check-cover-lb}
Any (randomized) algorithm for deciding whether a given $k = \Omega(\log n)$ sets covers all elements with probability of success at least $0.9$, requires $\Omega(nk)$ queries.
\end{theorem}

While this lower bound does not directly lead to a lower bound on \probSetCover, it suggests that verifying the feasibility of a solution may even be more costly than finding the approximate solution itself; any algorithm bypassing this $\Omega(nk)$ lower bound may not solve \probCoverVerification as a subroutine. 

We prove our lower bound by designing the \Yes and \No instances that are hard to distinguish, such that for a \Yes instance, the union of the given $k$ sets is $\sE$, while for a \No instance, their union only covers $n-1$ elements. Each \Yes instance is indistinguishable from a good fraction of \No instances. Thus any algorithm must unavoidably answer incorrectly on half of these fractions, and fail to reach the desired probability of success. 

\subsection{Underlying Set Structure.}
Our instance contains $n$ sets and $n$ elements (so $m = n$), where the first $k$ sets forms $\sS_k$, the candidate for the set cover we wish to verify. 
We first consider the incidence matrix representation, such that the rows represent the sets and the columns represent the elements. 
We focus on the first $n/k$ elements, and consider a \emph{slab}, composing of $n/k$ columns of the incidence matrix. 
We define a \emph{basic slab} as the structure illustrated in Figure~\ref{fig:basic_slab_swap} (for $n=12$ and $k = 3$), where the cell $(i, j)$ is {\em white} if $e_j \in S_i$, and is {\em gray} otherwise. 
The rows are divided into blocks of size $k$, where first block, the \emph{query block}, contains the rows whose sets we wish to check for coverage; notice that only the last element is not covered. 
More specifically, in a basic slab, the query block contains sets $S_1, \ldots, S_{n/k}$, each of which is equal to $\{e_1, \ldots, e_{n/k-1}\}$. The subsequent rows form the \emph{swapper blocks} each consisting of $n/k$ sets. The $r^{\textrm{th}}$ swapper block consists of sets $S_{(r+1)n/k + 1}, \ldots, S_{(r+2)n/k}$, each of which is equal to $\{e_1, \ldots, e_{n/k}\} \setminus \{e_{r}\}$.
\afterpage{%
\begin{figure}[!h]
    \centering
    \begin{subfigure}[b]{0.55\textwidth}
        \centering
        \begin{tikzpicture}[scale = 0.65]
\pgfmathsetmacro{\r}{3}  
\pgfmathsetmacro{\c}{4}  

\pgfmathsetmacro{\rr}{\r - 1}
\pgfmathsetmacro{\cc}{\c - 1}
\pgfmathsetmacro{\t}{\r * \c}
\foreach \x in {0}
	\fill[lightgray] (\x * \c + \cc, \r * \cc) rectangle (\x * \c + \cc + 1, \r * \cc + \r);
\foreach \x in {0}
	\foreach \y in {1, ..., \cc}
		\fill[lightgray] (\x * \c + \c - \y - 1, \r * \y - \r) rectangle (\x * \c + \c - \y, \r * \y);
\draw[step=1cm, thin] (0,0) grid (\c, \t);
\draw[line width = 0.8mm] (0,0) rectangle (\c, \t);
\foreach \x in {1}
	\draw[line width = 0.6mm] (\x * \c - \c, 0) rectangle (\x * \c, \t);

\foreach \y in {0, ..., \cc}
	\draw[line width = 0.6mm] (0, \y * \r) rectangle (\c, \y * \r + \r);

\draw[decoration={brace, mirror}, decorate] (-1.5,\t) -- node[left=2.5mm] {query block} (-1.5,\t - \r);
\foreach \y in {1, ..., \cc}{
	\draw[decoration={brace, mirror}, decorate] (-1.5,\t - \y * \r - 0.05) -- node[left=2.5mm] {swapper block \y} (-1.5,\t - \y * \r - \r + 0.05);
}
\foreach \x in {1, ..., \c}
	\node at (\x-0.5, 12.5) {$e_{\x}$};
\foreach \y in {1, ..., \t}
	\node at (-0.7, \t - \y + 0.5) {$S_{\y}$};
        \end{tikzpicture}
        \caption{a basic slab}
        \label{fig:slab_basic}
    \end{subfigure}
    \begin{subfigure}[b]{0.4\textwidth}
        \centering
        \begin{tikzpicture}[scale = 0.65]
\pgfmathsetmacro{\r}{3}  
\pgfmathsetmacro{\c}{4}  

\pgfmathsetmacro{\rr}{\r - 1}
\pgfmathsetmacro{\cc}{\c - 1}
\pgfmathsetmacro{\t}{\r * \c}
\foreach \x in {0}
	\fill[lightgray] (\x * \c + \cc, \r * \cc) rectangle (\x * \c + \cc + 1, \r * \cc + \r);
\foreach \x in {0}
	\foreach \y in {1, ..., \cc}
		\fill[lightgray] (\x * \c + \c - \y - 1, \r * \y - \r) rectangle (\x * \c + \c - \y, \r * \y);
\draw[step=1cm, thin] (0,0) grid (\c, \t);
\draw[line width = 0.8mm] (0,0) rectangle (\c, \t);
\foreach \x in {1}
	\draw[line width = 0.6mm] (\x * \c - \c, 0) rectangle (\x * \c, \t);
\foreach \y in {0, ..., \cc}
	\draw[line width = 0.6mm] (0, \y * \r) rectangle (\c, \y * \r + \r);
\foreach \x in {1, ..., \c}
	\node at (\x-0.5, 12.5) {$e_{\x}$};
\foreach \y in {1, ..., \t}
	\node at (-0.8, \t - \y + 0.5) {$S_{\y}$};

\foreach \p/\x/\y in {1/3/2}{
	\pgfmathsetmacro{\pp}{\p - 1}
	\pgfmathsetmacro{\xx}{\x - 1}
	\pgfmathsetmacro{\yy}{\y - 1}
	\draw[line width = 1mm, red, fill = lightgray, rounded corners = 0.5mm] (\pp * \c + \yy, \t - \xx) rectangle (\pp * \c + \yy + 1, \t - \xx - 1);
	\draw[line width = 1mm, red, fill = white, rounded corners = 0.5mm] (\p * \c - 1, \t - \xx) rectangle (\p * \c, \t - \xx - 1);
	\draw[line width = 1mm, red, fill = white, rounded corners = 0.5mm] (\pp * \c + \yy, \t - \y * \r - \xx) rectangle (\pp * \c + \yy + 1, \t - \y * \r - \xx - 1);
	\draw[line width = 1mm, red, fill = lightgray, rounded corners = 0.5mm] (\p * \c - 1, \t - \y * \r - \xx) rectangle (\p * \c, \t - \y * \r - \xx - 1);
}

\pgfmathsetmacro{\w}{0.15}
\draw[blue, dashed, very thick, rounded corners=\w*1cm] ( - \w, \t + \w) rectangle (\c - 1 + \w, \t - \r - \w);
        \end{tikzpicture}
        \caption{the slab after performing a $(3, 2)$-swap}
        \label{fig:slab_swapped}
    \end{subfigure}
	\caption{A basic slab and an example of a swapping operation.}
	\label{fig:basic_slab_swap}
\end{figure}
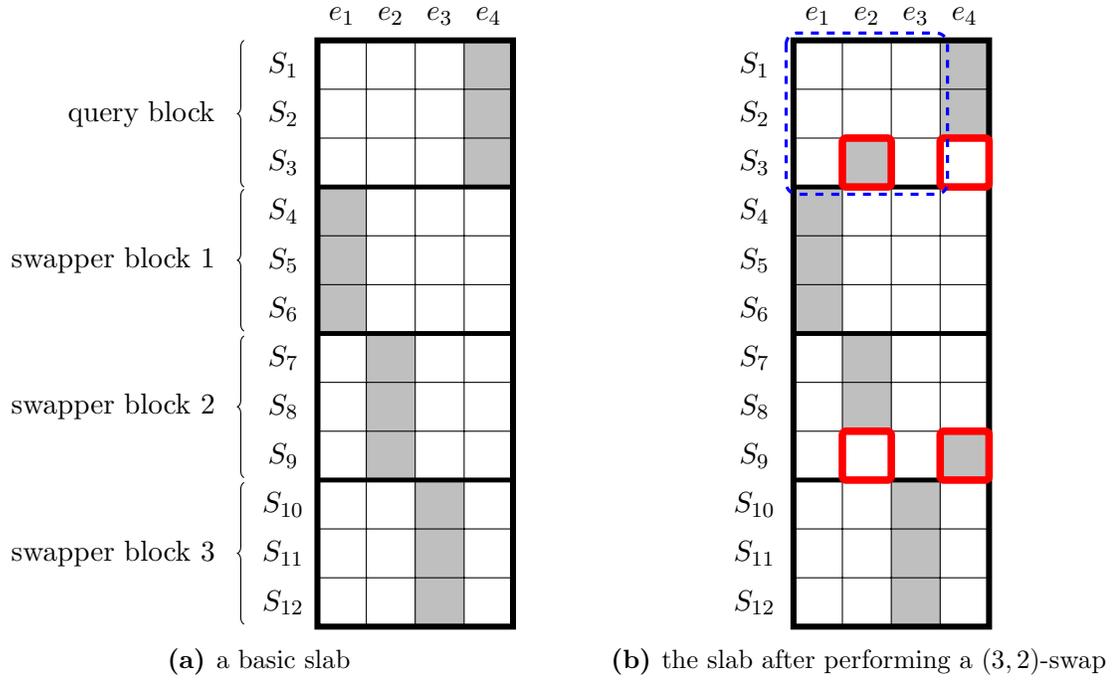

\vspace{10mm}

\begin{figure}[!h]
    \centering
    \begin{tikzpicture}[scale = 0.65]
\pgfmathsetmacro{\r}{3}  
\pgfmathsetmacro{\c}{4}  

\pgfmathsetmacro{\rr}{\r - 1}
\pgfmathsetmacro{\cc}{\c - 1}
\pgfmathsetmacro{\t}{\r * \c}
\foreach \x in {0, ..., \rr}
	\fill[lightgray] (\x * \c + \cc, \r * \cc) rectangle (\x * \c + \cc + 1, \r * \cc + \r);
\foreach \x in {0, ..., \rr}
	\foreach \y in {1, ..., \cc}
		\fill[lightgray] (\x * \c + \c - \y - 1, \r * \y - \r) rectangle (\x * \c + \c - \y, \r * \y);
\draw[step=1cm, thin] (0,0) grid (\t, \t);
\draw[line width = 0.8mm] (0,0) rectangle (\t, \t);
\foreach \x in {1, ..., \r}{
	\draw[line width = 0.6mm] (\x * \c - \c, 0) rectangle (\x * \c, \t);
	\draw[decoration={brace},decorate] (\x * \c - \c + 0.05,\t+1) -- node[above=2.5mm] {slab \x} (\x * \c - 0.05,\t+1);
}

\foreach \y in {0, ..., \cc}
	\draw[line width = 0.6mm] (0, \y * \r) rectangle (\t, \y * \r + \r);
\foreach \x in {1, ..., \t}
	\node at (\x-0.5, 12.5) {$e_{\x}$};
\foreach \y in {1, ..., \t}
	\node at (-0.8, \t - \y + 0.5) {$S_{\y}$};

\foreach \p/\x/\y in {1/3/2, 2/1/2, 3/1/3}{
	\pgfmathsetmacro{\pp}{\p - 1}
	\pgfmathsetmacro{\xx}{\x - 1}
	\pgfmathsetmacro{\yy}{\y - 1}
	\draw[line width = 1mm, red, fill = lightgray, rounded corners = 0.5mm] (\pp * \c + \yy, \t - \xx) rectangle (\pp * \c + \yy + 1, \t - \xx - 1);
	\draw[line width = 1mm, red, fill = white, rounded corners = 0.5mm] (\p * \c - 1, \t - \xx) rectangle (\p * \c, \t - \xx - 1);
	\draw[line width = 1mm, red, fill = white, rounded corners = 0.5mm] (\pp * \c + \yy, \t - \y * \r - \xx) rectangle (\pp * \c + \yy + 1, \t - \y * \r - \xx - 1);
	\draw[line width = 1mm, red, fill = lightgray, rounded corners = 0.5mm] (\p * \c - 1, \t - \y * \r - \xx) rectangle (\p * \c, \t - \y * \r - \xx - 1);
}
    \end{tikzpicture}
    \caption{A example structure of a \Yes instance; all elements are covered by the first $3$ sets.}
    \label{fig:slab_construction}
\end{figure}
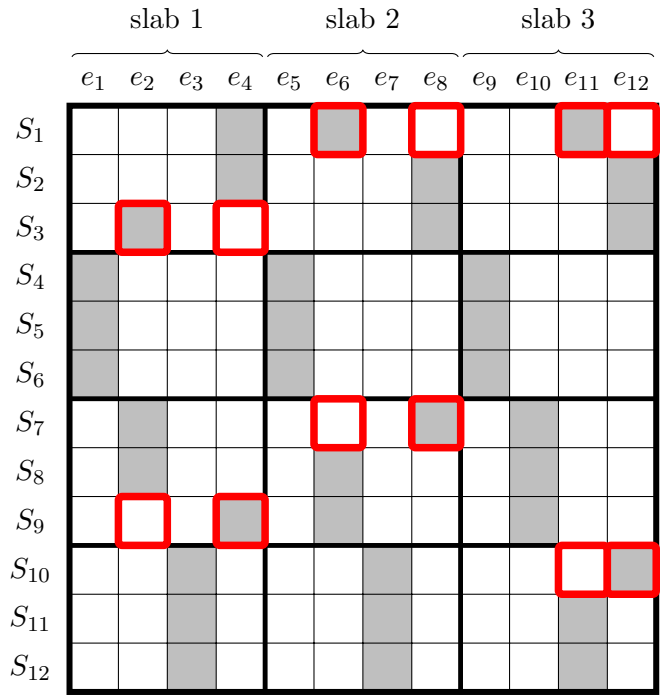

\clearpage
}
We perform one swap in this slab. Consider a parameter $(x, y)$ representing the index of a white cell within the query block. We exchange the color of this white cell with the gray cell on the same row, and similarly exchange the same pair of cells on swapper block $y$. An example is given in Figure~\ref{fig:basic_slab_swap}; the dashed blue rectangle corresponds to the indices parameterizing possible swaps, and the red squares mark the modified cells. This modification corresponds to a single $\swap$ operation; in this example, choosing the index $(3, 2)$ swaps $(e_2, e_4)$ between $S_3$ and $S_9$. Observe that there are $k \times (n/k  - 1) = n-k$ possible swaps on a single slab, and any single swap allows the query sets to cover all $n/k$ elements.

Lastly, we may create the full instance by placing all $k$ slabs together, as shown in Figure~\ref{fig:slab_construction}, shifting the elements' indices as necessary. The structure of our sets may be specified solely by the swaps made on these slabs. We define the structure of our instances as follows.
\begin{compactitem}
\item For a \Yes instance, we make one random swap on each slab. This allows the first $k$ sets to cover all elements.
\item For a \No instance, we make one random swap on each slab except for exactly one of them. In that slab, the last element is not covered by any of the first $k$ sets.
\end{compactitem}


Now, to properly define an instance, we must describe our structure via $\EltOf$ and $\SetOf$. We first create a temporary instance consisting of $k$ basic slabs, where none of the cells are swapped. Create $\EltOf$ and $\SetOf$ lists by sorting each list in an increasing order of indices. Each instance from the above construction can then be obtained by applying up to $k$ swaps on this temporary instance.
Figure~\ref{fig:slab_sample} provides a sample realization of a basic slab with $\EltOf$ and $\SetOf$, as well as a sample result of applying a swap on this basic slab; these correspond to the incidence matrices in Figure~\ref{fig:slab_basic} and Figure~\ref{fig:slab_swapped}, respectively. Such a construction can be extended to include all $k$ slabs. Observe here that no two distinct swaps modify the same entry; that is, the swaps do not interfere with one another on these two functions. We also note that many entries do not participate in any swap.

\afterpage{%
\renewcommand{\arraystretch}{1.5}
\begin{figure}[!h]
\centering
\begin{tabular}{p{1cm} c|c c c p{1cm} p{1cm} c|c c c p{1cm}}
  \multicolumn{6}{c}{\textbf{Before:} $\EltOf$ table for a basic slab} &\multicolumn{6}{c}{\textbf{After:} $\EltOf$ table after applying a swap}\\
  & $\EltOf$ & $1$ & $2$ & $3$ & & & $\EltOf$ & $1$ & $2$ & $3$ & \\ \hhline{~----~~----~}
  & $S_1$ & $e_1$ & $e_2$ & $e_3$ & & & $S_1$ & $e_1$ & $e_2$ & $e_3$ & \\ 
  & $S_2$ & $e_1$ & $e_2$ & $e_3$ & & & $S_2$ & $e_1$ & $e_2$ & $e_3$ & \\ 
  & $S_3$ & $e_1$ & $e_2$ & $e_3$ & & & $S_3$ & $e_1$ & \cellcolor{red}$e_4$ & $e_3$ & \\ 
  & $S_4$ & $e_2$ & $e_3$ & $e_4$ & & & $S_4$ & $e_2$ & $e_3$ & $e_4$ \\ 
  & $S_5$ & $e_2$ & $e_3$ & $e_4$ & & & $S_5$ & $e_2$ & $e_3$ & $e_4$ \\ 
  & $S_6$ & $e_2$ & $e_3$ & $e_4$ & & & $S_6$ & $e_2$ & $e_3$ & $e_4$ \\ 
  & $S_7$ & $e_1$ & $e_3$ & $e_4$ & & & $S_7$ & $e_1$ & $e_3$ & $e_4$ \\ 
  & $S_8$ & $e_1$ & $e_3$ & $e_4$ & & & $S_8$ & $e_1$ & $e_3$ & $e_4$ \\ 
  & $S_9$ & $e_1$ & $e_3$ & $e_4$ & & & $S_9$ & $e_1$ & $e_3$ & \cellcolor{red}$e_2$ \\ 
  & $S_{10}$ & $e_1$ & $e_2$ & $e_4$ & & & $S_{10}$ & $e_1$ & $e_2$ & $e_4$ \\ 
  & $S_{11}$ & $e_1$ & $e_2$ & $e_4$ & & & $S_{11}$ & $e_1$ & $e_2$ & $e_4$ \\ 
  & $S_{12}$ & $e_1$ & $e_2$ & $e_4$ & & & $S_{12}$ & $e_1$ & $e_2$ & $e_4$ \\ \\
\end{tabular}
\begin{tabular}{c|c c c c c c c c c}
 \multicolumn{10}{c}{\textbf{Before:} $\SetOf$ table for a basic slab} \\
 $\SetOf$ & $1$ & $2$ & $3$ & $4$ & $5$ & $6$ & $7$ & $8$ & $9$ \\ \hhline{----------}
 $e_1$ & $S_1$ & $S_2$ & $S_3$ & $S_7$ & $S_8$ & $S_9$ & $S_{10}$ & $S_{11}$ & $S_{12}$ \\
 $e_2$ & $S_1$ & $S_2$ & $S_3$ & $S_4$ & $S_5$ & $S_6$ & $S_{10}$ & $S_{11}$ & $S_{12}$ \\
 $e_3$ & $S_1$ & $S_2$ & $S_3$ & $S_4$ & $S_5$ & $S_6$ & $S_7$ & $S_8$ & $S_9$ \\
 $e_4$ & $S_4$ & $S_5$ & $S_6$ & $S_7$ & $S_8$ & $S_9$ & $S_{10}$ & $S_{11}$ & $S_{12}$ \\ \multicolumn{10}{c}{} \\
\end{tabular}
\begin{tabular}{c|c c c c c c c c c}
 \multicolumn{10}{c}{\textbf{After:} $\SetOf$ table after applying a swap} \\
 $\SetOf$ & $1$ & $2$ & $3$ & $4$ & $5$ & $6$ & $7$ & $8$ & $9$ \\ \hhline{----------}
 $e_1$ & $S_1$ & $S_2$ & $S_3$ & $S_7$ & $S_8$ & $S_9$ & $S_{10}$ & $S_{11}$ & $S_{12}$ \\
 $e_2$ & $S_1$ & $S_2$ & \cellcolor{red}$S_9$ & $S_4$ & $S_5$ & $S_6$ & $S_{10}$ & $S_{11}$ & $S_{12}$ \\
 $e_3$ & $S_1$ & $S_2$ & $S_3$ & $S_4$ & $S_5$ & $S_6$ & $S_7$ & $S_8$ & $S_9$ \\
 $e_4$ & $S_4$ & $S_5$ & $S_6$ & $S_7$ & $S_8$ & \cellcolor{red}$S_3$ & $S_{10}$ & $S_{11}$ & $S_{12}$ \\
\end{tabular}
\vspace{5mm}
\caption{Tables illustrating the representation of a slab under $\EltOf$ and $\SetOf$ before and after a swap; cells modified by $\swap(e_2, e_4) $ between $S_3$ and $S_9$ are highlighted in red.}
\label{fig:slab_sample}
\end{figure}
\clearpage
}

\subsection{Proof of Theorem \ref{thm:check-cover-lb}.}

Observe that according to our instance construction, the algorithm may verify, with a single query, whether a certain swap occurs in a certain slab. Namely, it is sufficient to query an entry of $\EltOf$ or $\SetOf$ that would have been modified by that swap, and check whether it is actually modified or not. For simplicity, we assume that the algorithm has the knowledge of our construction. Further, without loss of generality, the algorithm does not make multiple queries about the same swap, or make a query that is not corresponding to any swap. 

We employ Yao's principle as follows: to prove a lower bound for randomized algorithms, we show a lower bound for any deterministic algorithm on a fixed distribution of input instances. Let $s = n-k$ be the number of possible swaps in each slab; assume $s = \Theta(n)$. We define our distribution of instances as follows: each of the $s^k$ possible \Yes instances occurs with probability $1/(2s^k)$, and each of the $ks^{k-1}$ possible \No instances occurs with probability $1/(2ks^{k-1})$. Equivalently speaking, we create a random \Yes instance by making one swap on each basic slab. Then we make a coin flip: with probability $1/2$ we pick a random slab and undo the swap on that slab to obtain a \No instance; otherwise we leave it as a \Yes instance. To prove by contradiction, assume there exists a deterministic algorithm that solves the \probCoverVerification problem over this distribution of instances with $r = o(sk)$ queries.

Consider the \Yes instances portion of the distribution, and observe that we may alternatively interpret the random process generating them as as follows. For each slab, one of its $s$ possible swaps is chosen uniformly at random. This condition again follows the scenario considered in Section~\ref{sec:ind-reduction}: we are given $k$ urns (slabs) of each consisting of $s$ marbles (possible swap locations), and aim to draw the $\red$ marble (swapped entry) from a large fraction of these urns. Following the proof of Lemmas~\ref{lem:high-threshold-slabs}-\ref{lem:sim-output}, we obtain that if the total number of queries made by the algorithm is less than $(1-{3\over b}){sk \over b}$, then with probability at least $0.99$, the algorithm will not see any swaps from at least ${k\over b}$ slabs.

Then, consider the corresponding \No instances obtained by undoing the swap in one of the slabs of the \Yes instance. Suppose that the deterministic algorithm makes less than $(1-{3\over b}){sk \over b}$ queries, then for a fraction of $0.99$ of all possible tuples $\sT$, the output of the \Yes instance is the same as the output of ${1\over b}$ fraction of \No instances, namely when the slab containing no swap is one of the ${k \over b}$ slabs that the algorithm has not detected a swap in the corresponding \Yes instance; the algorithm must answer incorrectly on half of the corresponding weight in our distribution of input instances. Thus the probability of success for any algorithm with less than $(1-{3\over b}){sk \over b}$ queries is at most
\begin{align*}
1-\Prob{|\sT_{\mathrm{high}}| \geq (1-{2\over b})k} ({1 \over b}) ({1 \over 2}) \leq 1 - {0.495\over b} < 0.9,
\end{align*} 
for a sufficiently small constant $b > 3$ (e.g. $b=4$). As $s = \Theta(n)$ and by Yao's principle, this implies the lower bound of $\Omega(nk)$ for the \probCoverVerification problem.

\section*{Acknowledgment}
We would like to thank Jonathan Ullman for many helpful discussions.

\appendix
\section{Omitted Proofs from Section~\ref{sec:simple-lb}}\label{sec:missing-proofs}

\begin{lemma}\label{lem:solution-size}
With probability at least $1-m^{-1}$ over $\sS\sim \binstF(\sE,p_0)$, the size of the minimum set cover of the instance $(\sS, \sE)$ is greater than $2$.
\end{lemma}
\begin{proof}
The probability that an element $e\in\sE$ is covered by two sets selected from $\sS$ is at most:
\begin{align*}
\pr[e  \in S_1\cup S_2 ] = 1 - p_{0}^2 = 1-{9\log m \over n}.
\end{align*} 
Thus, the probability that $S_1 \cup S_2$ covers all elements in $\sE$ is at most $(1-{9\log m \over n})^n < m^{-9}$.
Applying the union bound, with probability at least $1-m^{-1}$ the size of optimal set cover is greater than $2$. 
\end{proof}

\begin{lemma}\label{lem:missing_elem_size}
Let $S_1$ and $S_2$ be two sets in $\sS$ where $\sS\sim \binstF(\sE,p_0)$. Then with probability at least $1-m^{-1}$, $\card{\sE\setminus(S_1 \cup S_2)} \leq 18\log m$.
\end{lemma}      
\begin{proof}
For an element $e$, $\pr[ e\notin S_1 \cup S_2] = p_0^2 = {9\log m \over n}$. So, $\mathbb{E}[\card{\sE\setminus(S_1 \cup S_2)}] = 9\log m$. By Chernoff bound, $\pr[ \card{\sE\setminus(S_1 \cup S_2)}\geq 18 \log m]$ is at most $e^{-9\log m/3}\leq m^{-3}$. Thus with probability at least $1-m^{-1}$, for any pair of sets in $\sS$, the number of element not covered by their union is at most $18\log m$. 
\end{proof}

\begin{lemma}\label{lem:intersection_size}
Let $S_1$ and $S_2$ be two sets in $\sS$ where $\sS \sim \binstF(\sE, p_0)$. Then $\card{S_1 \cap S_2} \geq n/8$ with probability at least $1-m^{-1}$.
\end{lemma}      
\begin{proof}
For each element $e$, it is either covered by both $S_1,S_2$, one of $S_1, S_2$ or none of them. Since $p_0\leq 1/2$, the probability that an element is covered by both sets is greater than other cases, i.e., $\pr\left[ e \in S_1 \cap S_2 \right] > 1/4$. Thus, $\mathbb{E}[\card{\sE\setminus(S_1 \cap S_2)}] > n/4$. By Chernoff bound, $\pr[ \card{\sE\setminus(S_1 \cap S_2)} \leq n/8]$ is exponentially small. Thus with probability at least $1-m^{-1}$, the intersection of any pairs of sets in $\sS$ is greater than $n/8$.
\end{proof}

\begin{lemma}\label{lem:candidate-swapping}
Suppose that $\sS \sim \binstF(\sE, p_0)$ and let $e, e'$ be two elements in $\sE$. With probability at least $1-m^{-1}$, the number of  sets $S\in \sS$ such that $e\in S$ but $e'\notin S$ is at least ${m\sqrt{9\log m} \over 4\sqrt{n}}$. 
\end{lemma}            
\begin{proof}
For each set $S$, $\pr[e\in S \text{ and } e'\notin S] = (1-p_0)p_0 \geq p_0/2$. This implies that the expected number of $S$ satisfying the condition for $e$ and $e'$ is at least ${m\over 2}\cdot{\sqrt{9\log m \over n}}$ and by Chernoff bound, the probability that the number of sets containing $e$ but not $e'$ is less than ${m\sqrt{9\log m} \over 4\sqrt{n}}$ is exponentially small. Thus with probability at least $1-m^{-1}$  property (\ref{item:mi-candidate-size}) holds for any pair of elements in $\sE$. 
\end{proof}

\begin{lemma}\label{lem:triple-1}
Suppose that $\sS \sim \binstF(\sE, p_0)$ and let $S_1, S_2$ and $S$ be sets in $\sS$. With probability at least $1-n^{-1}$, $\card{(S_1 \cap S_2) \setminus S} \leq 6\sqrt{n\log m}$. 
\end{lemma}            
\begin{proof}
For each element $e$, $\pr[e\in (S_1 \cap S_2) \setminus S] = (1-p_0)^2p_0 \leq p_0$. This implies that the expected size of $(S_1\cap S_2)\setminus S$ is less than $\sqrt{9n\log m}$ and by Chernoff bound, the probability that $\card{(S_1\cap S_2)\setminus S} \geq 6\sqrt{n\log m}$ is exponentially small. Thus with probability at least $1-m^{-1}$ property $(\ref{item:mi-triple})$ holds for any sets $S_1, S_2$ and $S$ in $\sS$. 
\end{proof}

\begin{lemma}\label{lem:element-deg}
For each element, the number of sets that do not contain the element is at most $6 m\sqrt{\log m \over n}$.
\end{lemma}
\begin{proof}
For each element $e$, $\pr_S[e\notin S] = p_0$. This implies that $\mathbb{E}_S(\card{\set{S \sep e\notin S}})$ is less than $m\sqrt{9\log m \over n}$ and by Chernoff bound, the probability that $\card{\set{S \sep e\notin S}} \geq 2m\sqrt{9\log m \over n}$ is exponentially small. Thus with probability at least $1-m^{-1}$  property (\ref{item:mi-degneg}) holds for any element $e\in \sE$. 
\end{proof}


\section{Generalized Lower Bounds for the Set Cover Problem}\label{sec:general-lb}
In this section we generalize the approach of Section~\ref{sec:simple-lb} and prove our main lower bound result (Theorem~\ref{thm:lowerbound-general}) for the number of queries required for approximating with factor $\alpha$ the size of an optimal solution to the \probSetCover problem, where the input instance contains $m$ sets, $n$ elements, and a minimum set cover of size $k$. The structure of our proof is largely the same as the simplified case, but the definitions and the details of our analysis will be more complicated. The size of the minimum set cover of the median instance will instead be at least $\alpha k + 1$, and $\genModifiedInst$ reduces this down to $k$. We now aim to prove the following statement which implies the lower bound in Theorem~\ref{thm:lowerbound-general}. 

\begin{theorem} \label{thm:general-lb-2}
Let $k$ be the size of an optimal solution of $\binst$ such that $1< \alpha \leq \log n$ and $2\leq k\leq \left(\frac{n}{16\apxF\log m}\right)^{1 \over 4\apxF+1}$. 
Any algorithm that distinguishes whether the input instance is $\binst$ or belongs to $\minstD(\binst)$ with probability of success at least $2/3$ requires $\tldOmega(m({n \over k})^{1/(2\apxF)})$ queries.
\end{theorem}

\subsection{Construction of the Median Instance $\binst$.}
Let $\sS$ be a collection of $m$ sets such that independently for each set-element pair $(S, e)$, $S$ contains $e$ with probability $1-p_0$, where we modify the probability to $p_0 = \left({8(\apxF k+2) \log m \over n}\right)^{1/(\apxF k)}$. We start by proving some inequalities involving $p_0$ that will be useful later on, which hold for any $k$ in the assumed range.
\begin{lemma}\label{lem:gen-k-property}
For $2 \leq k\leq \left(\frac{n}{16\apxF\log m}\right)^{1 \over 4\apxF+1}$, we have that
\begin{compactenum}[(a)]
\item $1 - p_0 \geq p_0^{k/4}$,
\item $p_0^{k/4} \leq 1/2$,
\item $\frac{p_0^k}{(1-p_0)^2}\leq \left(\frac{8(\apxF k+2)\log m}{n}\right)^{1 \over 2\apxF}$.
\end{compactenum}
\end{lemma}
\begin{proof}
Recall as well that $\alpha > 1$. In the given range of $k$, we have $k^{4\apxF} \leq \frac{n}{16\apxF k \log m} \leq \frac{n}{8(\apxF k+2)\log m}$ because $k\apxF\geq 2$. Thus
\begin{align*}
p_0=\left(\frac{8(\apxF k+2)\log m}{n}\right)^{1 \over \apxF k} \leq \left(\frac{1}{k^{4\apxF}}\right)^{1 \over \apxF k} = k^{-4 /k}.
\end{align*}
Next, rewrite $k^{-4 /k} = e^{-{{4 \ln k} \over k}}$ and observe that ${{4 \ln k} \over k} \leq {4 \over e} < 1.5$. Since $e^{-x} \leq 1-{x \over 2}$ for any $x<1.5$, we have $p_0 \leq e^{-{{4 \ln k} \over k}} < 1-\frac{2\ln k}{k}$. Further, $p_0^{k/4} \leq e^{-\ln k} =1/k$. Hence $p_0+p_0^{k/4} \leq 1-\frac{2\ln k}{k} + \frac{1}{k} \leq 1$, implying the first statement.

The second statement easily follows as $p_0^{k/4} \leq 1/k \leq 1/2$ since $k \geq 2$. For the last statement, we make use of the first statement:
\begin{align*}
\frac{p_0^k}{(1-p_0)^2}\leq \frac{p_0^k}{(p_0^{k/4})^2} =p_0^{k/2} = \left(\frac{8(\apxF k +2)\log m}{n}\right)^{1 \over 2\apxF}
\end{align*}
which completes the proof of the lemma.
\end{proof}

Next, we give the new, generalized definition of median instances.

\begin{definition}[Median instance]\label{def:general-median-instance}
An instance of \prob{Set Cover}, $\inst = (\sE, \sS)$, is a \emph{median instance} if it satisfies all the following properties.
\begin{compactenum}[(a)]
\item No $\apxF k$ sets cover all the elements. (The size of its minimum set cover is greater than $\apxF k$.)\label{item:gen-sc-size}
\item The number of uncovered elements of the union of any $k$ sets is at most ${2np_0^k}$.\label{item:gen-uncovered-elem-size}
\item For any pair of elements $e, e'$, the number of sets $S\in \sS$ s.t. $e\in S$ but $e'\notin S$ is at least ${(1-p_0)p_0 m/2}$.\label{item:gen-candidate-size}
\item For any collection of $k$ sets $S_1, \cdots, S_k$, $\card{S_k \cap (S_1\cup \cdots \cup S_{k-1})} \geq (1-p_0)(1-p_0^{k-1})n/2$.\label{item:gen-s-candidate-element}
\item For any collection of $k+1$ sets $S,S_1,\cdots, S_k$, $\card{(S_k \cap (S_1\cup\cdots \cup S_{k-1})) \setminus S} \leq 2p_0(1-p_0)(1-p_0^{k-1}) n$.\label{item:gen-candidate-element}
\item For each element, the number of sets that do not contain the element is at most $(1+{1\over k})p_0m$.\label{item:gen-degneg}
\end{compactenum}
\end{definition}

\begin{lemma}
For $k\leq \min\{\sqrt{\frac{m}{27 \ln m}},(\frac{n}{16\apxF\log m})^{1 \over 4\apxF+1}\}$, there exists a median instance $\binst$ satisfying all the median properties from Definition~\ref{def:general-median-instance}. In fact, most of the instances constructed by the described randomized procedure satisfy the median properties.
\end{lemma}

\begin{proof}
The lemma follows from applying the union bound on the results of Lemmas~\ref{lem:gen-solution-size}--\ref{lem:gen-element-degneg}.
\end{proof}

The proofs of the Lemmas~\ref{lem:gen-solution-size}--\ref{lem:gen-element-degneg} follow from standard applications of concentration bounds.  We include them here for the sake of completeness.

\begin{lemma}\label{lem:gen-solution-size}
With probability at least $1-m^{-2}$ over $\sS\sim \binstF(\sE,p_0)$, the size of the minimum set cover of the instance $(\sS, \sE)$ is at least $\apxF k+1$.
\end{lemma}
\begin{proof}
The probability that an element $e\in\sE$ is covered by a specific collection of $\apxF k$ sets in $\sS$ is at most
$1 - p_{0}^{\apxF k} = 1-{8(\apxF k + 2)\log m \over n}$.
Thus, the probability that the union of the $\apxF k $ sets covers all elements in $\sE$ is at most $(1-{8(\apxF k + 2)\log m \over n})^n < m^{-8(\alpha k + 2)}$.
Applying the union bound, with probability at least $1-m^{-2}$ the size of an optimal set cover is at least $\apxF k +1$. 
\end{proof}

\begin{lemma}\label{lem:gen-missing_elem_size}
With probability at least $1-m^{-2}$ over $\sS\sim \binstF(\sE,p_0)$, any collection of $k$ sets has at most $2np_0^{k}$ uncovered elements.
\end{lemma}      
\begin{proof}
Let $S_1,\cdots,S_k$ be a collection of $k$ sets from $\sS$. For each element $e\in \sE$, the probability that $e$ is not covered by the union of the $k$ sets is $p_0^k$. Thus, 
\begin{align*}
\mathbb{E}[\card{\sE\setminus(S_1 \cup \cdots \cup S_k)}] = p_0^k n \geq p_0^{\apxF k} n = 8(\apxF k+2)\log m.
\end{align*} 
By Chernoff bound, 
\begin{align*}
\Prob{\card{\sE\setminus(S_1 \cup \cdots \cup S_k)}\geq 2p_0^k n} &\leq e^{-{p_0^kn \over 3}} \leq e^{-(\apxF k +2)\log m} \leq m^{-k-2}.
\end{align*} 
Thus with probability at least $1-m^{-2}$, for any collection of $k$ sets in $\sS$, the number of uncovered elements by the union of the sets is at most $2p_0^k n$. 
\end{proof}

\begin{lemma}\label{lem:gen-candidate-swapping}
Suppose that $\sS \sim \binstF(\sE, p_0)$ and let $e, e'$ be two elements in $\sE$. Given $k\leq \left(\frac{n}{16\apxF\log m}\right)^{1 \over 4\apxF+1}$, with probability at least $1-m^{-2}$, the number of  sets $S\in \sS$ such that $e\in S$ but $e'\notin S$ is at least ${mp_0(1-p_0)/2}$. 
\end{lemma}
\begin{proof}
For each set $S$, $\Prob{e\in S \text{ and } e'\notin S} = (1-p_0)p_0$. This implies that the expected number of such sets $S$ satisfying the condition for $e$ and $e'$ is
\begin{align*}
{p_0(1-p_0)m} \geq p_0\cdot p_0^{k/4}\cdot m \geq p_0^{\apxF k} n = 8(\apxF k + 2) \log m
\end{align*} 
by Lemma~\ref{lem:gen-k-property} and $m\geq n$. By Chernoff bound, the probability that the number of sets containing $e$ but not $e'$ is less than ${mp_0(1-p_0)/2}$ is at most
\begin{align*}
e^{-{p_0(1-p_0)m \over 8}} \leq e^{-(\apxF k +2)\log m} \leq m^{-\apxF k-2}.
\end{align*} 
Thus with probability at least $1-m^{-2}$  property~(\ref{item:gen-candidate-size}) holds for any pair of elements in $\sE$. 
\end{proof}

\begin{lemma}\label{lem:gen-matching-element-size}
Suppose that $\sS \sim \binstF(\sE, p_0)$ and let $S_1, \cdots, S_k$ be $k$ different sets in $\sS$. Given $k\leq \left(\frac{n}{16\apxF\log m}\right)^{1 \over 4\apxF+1}$, with probability at least $1-m^{-2}$, $\card{S_k \cap (S_1\cup\cdots \cup S_{k-1})} \geq (1-p_0)(1-p_0^{k-1}) n/2$. 
\end{lemma}            
\begin{proof}
For each element $e$, $\Prob{e\in S_k \cap (S_1\cup\cdots \cup S_{k-1})} = (1-p_0)(1-p_0^{k-1})$. This implies that the expected size of $S_k \cap (S_1\cup\cdots \cup S_{k-1})$ is 
\begin{align*}
(1-p_0)(1-p_0^{k-1}) n \geq p_0^{k/4}\cdot p_0^{k/4} \cdot n &\geq p_0^{\apxF k} n = 8(\apxF k +2)\log m. 
\end{align*} 
by Lemma~\ref{lem:gen-k-property}. By Chernoff bound, the probability that $\card{S_k \cap (S_1\cup\cdots \cup S_{k-1})} \leq (1-p_0)(1-p_0^{k-1}) n/2$ is at most
\begin{align*}
e^{-{(1-p_0)(1-p_0^{k-1}) n \over 8}} \leq e^{-(\apxF k +2)\log m} \leq m^{-\apxF k-2}.
\end{align*} 
Thus with probability at least $1-m^{-2}$  property~(\ref{item:gen-s-candidate-element}) holds for any sets $S_1, \cdots, S_k$ in $\sS$. 
\end{proof}

\begin{lemma}\label{lem:gen-matching-swap-size}
Suppose that $\sS \sim \binstF(\sE, p_0)$ and let $S_1, \cdots, S_k$ and $S$ be $k+1$ different sets in $\sS$. Given $k\leq \left(\frac{n}{16\apxF\log m}\right)^{1 \over 4\apxF+1}$, with probability at least $1-m^{-2}$, $\card{ (S_k \cap (S_1\cup\cdots \cup S_{k-1}))\setminus S} \leq 2p_0(1-p_0)(1-p_0^{k-1})n$. 
\end{lemma}
\begin{proof}
For each element $e$, $\Prob{e\in (S_k \cap (S_1\cup\cdots \cup S_{k-1}))\setminus S} = p_0(1-p_0)(1-p_0^{k-1})$. Then, 
\begin{align*}
\mathbb{E}(\card{(S_k \cap (S_1\cup\cdots \cup S_{k-1}))\setminus S}) = p_0(1-p_0)(1-p_0^{k-1}) n \geq p_0\cdot p_0^{k/4} \cdot p_0^{k/4} &\geq p_0^{\apxF k} n \\
&= 8(\apxF k+2)\log m 
\end{align*} 
by Lemma \ref{lem:gen-k-property}. By Chernoff bound, the probability that $\card{(S_k \cap (S_1\cup\cdots \cup S_{k-1}))\setminus S} \geq 2p_0(1-p_0)(1-p_0^{k-1}) n$ is 
\begin{align*}
e^{-{p_0(1-p_0)(1-p_0^{k-1}) n \over 3}} \leq e^{-2(\apxF k +2)\log m} \leq m^{-2\apxF k-4}.
\end{align*}
Thus with probability at least $1-m^{-2}$  property~(\ref{item:gen-candidate-element}) holds for any sets $S_1, \cdots, S_k$ and $S$ in $\sS$. 
\end{proof}

\begin{lemma}\label{lem:gen-element-degneg}
Given that $k\leq \left(\frac{n}{16 \apxF \log m}\right)^{\frac{1}{4\apxF+1}}$, for each element, the number of sets that do not contain the element is at most $(1+{1\over k})p_0 m$. 
\end{lemma}
\begin{proof}
First, note that $k\leq \left(\frac{n}{16 \apxF \log m}\right)^{\frac{1}{4\apxF+1}}\leq\sqrt{\frac{m}{27\ln m}}$ as $m\geq n$ and $\apxF\geq 1$.

Next, for each element $e$, $\pr_{S\sim \sS}[e\notin S] = p_0$. This implies that 
$\mathbb{E}_S(\card{\set{S \sep e\notin S}}) = p_0 m$.
By Chernoff bound, the probability that $\card{\set{S \sep e\notin S}} \geq (1+{1\over k})p_0 m$ is at most $e^{-mp_0\over 3k^2}$.
Now if $k\geq \log n$, then $p_0 \geq 1/e$ and thus this probability would be at most $\exp({-m\over 3ek^2}) \leq m^{-3}$ for any $k\leq \sqrt{\frac{m}{27\ln m}}$. Otherwise, we have that the above probability is at most $\exp({-m n^{-1/\apxF k} \over 3\log^2 n}) \leq \exp({-m^{1-1/\apxF k} \over 3\log^2 m})\leq m^{-3}$ given $m\geq n$ and sufficiently large $n$.
Thus with probability at least $1-m^{-2}$  property~(\ref{item:gen-degneg}) holds for any element $e\in \sE$. 
\end{proof}

\subsection{Distribution $\minstD(\binst)$ of the Modified Instances Derived from $\binst$.}

Fix a median instance $\binst$. We now show that we may perform $\tilde O(n^{1-1/\alpha}k^{1/\apxF})$ $\swap$ operations on $\binst$ so that the size of the minimum set cover in the modified instance becomes $k$. So, the number of queries to $\EltOf$ and $\SetOf$ that induce different answers from those of $\binst$ is at most $\tldO(n^{1-1/\apxF}k^{1/\apxF})$. 
We define $\minstD(\binst)$ as the distribution of instances $\minst$ that is generated from a median instance $\binst$ by $\genModifiedInst(\binst)$ given below in Figure~\ref{gen:modified-instance}. 
The main difference from the simplified version are that we now select $k$ different sets to turn them into a set cover, and the swaps may only occur between $S_k$ and the candidates.

\begin{figure}[h]
    \begin{center}
        \begin{algorithmEnv}
            \underline{\genModifiedInst%
               $\pth{ \binst=\pth{\sE, \sS}}$}: $\Bigl.$ \qquad
            \+\\ %

            $\sM \leftarrow \emptyset$ \\ 
            {\bf pick} $k$ \emph{different} sets $S_1, \cdots S_k$ from $\sS$ uniformly at random\\
            \Foreach $e\in \sE\setminus (S_1 \cup \cdots \cup S_k)$ \Do\\
            \> {\bf pick} $e'\in (S_k\cap (S_1\cup\cdots \cup S_{k-1}))\setminus \sM$ uniformly at random \\
            \> $\sM \leftarrow \sM \cup \set{ee'}$\\
            \> {\bf pick} a random set $S$ in $\cand(e,e')$\\
            \> $\swap(e,e')$ between $S,S_k$
        \end{algorithmEnv}%
    \end{center}
    \vspace{-0.15in}
    \caption{The procedure of constructing a modified instance of $\binst$.  }
    \label{gen:modified-instance}
\end{figure}
\begin{lemma}
The procedure \emph{$\genModifiedInst$} is well-defined under the precondition that the input instance $\binst$ is a median instance.
\end{lemma}
\begin{proof}
To carry out the algorithm, we must ensure that the number of the initially uncovered elements is at most that of the elements covered by both $S_k$ and some other set from $S_1, \ldots, S_{k-1}$. Since $\binst$ is a median instance, by properties~(\ref{item:gen-uncovered-elem-size}) and~(\ref{item:gen-s-candidate-element}) from Definition~\ref{def:general-median-instance}, these values satisfy $\card{\sE\setminus(S_1\cup \cdots \cup S_k)} \leq 2p_0^k n$ and $\card{S_k \cap (S_1\cup\cdots\cup S_{k-1})} \geq (1-p_0)(1-p_0^{k-1})n/2$, respectively. By Lemma \ref{lem:gen-k-property}, $p_0^{k/4}\leq 1/2$. Using this and Lemma \ref{lem:gen-k-property} again,
\begin{align*}
(1-p_0)(1-p_0^{k-1})n/2 &\geq p_0^{k/4}\cdot p_0^{k/4} \cdot n/2 \geq  p_0^{k/2}n/2 \geq 2p_0^{k}n.
\end{align*}
That is, in our construction there are sufficiently many possible choices for $e'$ to be matched and swapped with each uncovered element $e$. Moreover, since $\binst$ is a median instance, $\card{\cand(e,e')} \geq {(1-p_0)p_0 m / 2}$ (by property~(\ref{item:gen-candidate-size})), and there are plenty of candidates for each swap.
\end{proof}

\subsubsection{Bounding the Probability of Modification.}\label{sec:gen-almost-uniform}

Similarly to the simplified case, define $P_{\eltq-\setq}:\sE\times\sS \rightarrow [0,1]$ as the probability that an element is swapped by a set, and upper bound it via the following lemma.

\begin{lemma}\label{lem:gen_p_cell_uniform}
For any $e \in \sE$ and $S \in \sS$, $P_{\eltq-\setq}(e,S)\leq \frac{64p_0^k}{(1-p_0)^2m}$ where the probability is taken over the random choices of $\minst\sim\minstD(\binst)$.
\end{lemma}

\begin{proof}
Let $S_1, \ldots, S_k$ denote the first $k$ sets picked (uniformly at random) from $\sS$ to construct a modified instance of $\binst$. For each element $e$ and a set $S$ such that $e\in S$ in the basic instance $\binst$,
\begin{align*}
&P_{\eltq-\setq}(e,S) =  \Prob{S = S_k} 
		\cdot \Prob{e\in \cup_{i\in[k-1]}S_i \sep e\in S_k}\\ 
		&\cdot \Prob{e \text{ matches to } \sE\setminus(\cup_{i\in[k]}S_i) \sep e\in S_k \cap (\cup_{i\in[k-1]}S_i)} \\
&+ \Prob{S\notin\set{S_1, \ldots, S_k}} 
		\cdot \Prob{e\in S\setminus (\cup_{i\in[k]}S_i) \sep e\in S}\\	
		&\cdot \Prob{S \text{ swaps } e \text{ with } S_k \sep e\in S\setminus (S_1\cup \cdots \cup S_k)},
\end{align*}
where all probabilities are taken over $\minst \sim \minstD(\binst)$.
Next we bound each of the above six terms. Clearly, since we choose the sets $S_1,\cdots, S_k$ randomly, $\pr[S=S_k] = 1/m$. We bound the second term by $1$. Next, by properties~(\ref{item:gen-uncovered-elem-size}) and~(\ref{item:gen-s-candidate-element}) of median instances, the third term is at most
\begin{align*}
\frac{\card{\sE\setminus (\cup_{i\in[k]}S_i)}}{\card{S_k\cap (\cup_{i\in[k-1]}S_i)}}
\leq \frac{2p_0^kn}{(1-p_0)(1-p_0^{k-1}){n\over 2}}
\leq \frac{4p_0^k}{(1-p_0)^2}.
\end{align*}
We bound the fourth term by $1$. Let $d_e$ denote the number of sets in $\sS$ that do not contain $e$. Using  property~(\ref{item:gen-degneg}) of median instances, the fifth term is at most
\begin{align*}
\frac{d_e (d_e-1)\cdots (d_e-k+1)}{(m-1)(m-2)\cdots (m-k)} &\leq \left( \frac{d_e}{m-1}\right)^{k} \leq (\frac{(1+1/k)p_0m}{m (1-\frac{1}{k+1})})^k \leq e^2 p_0^k, 
\end{align*}
Finally for the last term, note that by symmetry, each pair of matched elements $ee'$ is picked by $\genModifiedInst$ equiprobably. Thus, for any $e\in S\setminus (S_1 \cup \cdots\cup S_k)$, the probability that each element $e'\in  S_k \cap (S_1\cup\cdots\cup S_{k-1})$ is matched to $e$ is ${1 \over \card{ S_k \cap (S_1\cup\cdots\cup S_{k-1})}}$. By properties~(\ref{item:gen-candidate-size})-(\ref{item:gen-candidate-element}) of median instances, the last term is at most
\begin{align*}
&\sum_{e'\in (S_k \cap (\cup_{i\in[k-1]}S_i))\setminus S}  \Prob{ee'\in \sM}\cdot \Prob{(S,S_k) \text{ swap } (e,e')} \\
& \leq  \card{(S_k \cap (\cup_{i\in[k-1]}S_i))\setminus S}
\!\begin{aligned}[t]
	&\cdot {1 \over \card{S_k \cap (\cup_{i\in[k-1]}S_i)}}\cdot {1 \over \card{\cand(e,e')}} \\
\end{aligned}
\\
&\leq 2p_0(1-p_0)(1-p_0^{k-1})n 
\!\begin{aligned}[t]
	&\cdot \frac{1}{(1-p_0)(1-p_0^{k-1})n/2} \cdot \frac{1}{p_0(1-p_0)m/2} \\ 
\end{aligned}
\\
& \leq \frac{8}{(1-p_0)m}
\end{align*}
Therefore, 
\begin{align*}
P_{\eltq-\setq}(e,S) &\leq \frac{1}{m}\cdot 1 \cdot \frac{4p_0^k}{(1-p_0)^2} + 1\cdot e^2p_0^k \cdot \frac{8}{(1-p_0)m} \\&\leq \frac{4p_0^k}{(1-p_0)^2} + \frac{60p_0^k}{(1-p_0)m} \leq
\frac{64p_0^k}{(1-p_0)^2m}.
\end{align*}
\end{proof}

\subsection{Proof of Theorem \ref{thm:general-lb-2}.}\label{sec:sc-lowerbound-proof}
The remaining part of our proof follows that of the simplified version almost exactly.

\noindent
{\em Proof of Theorem~\ref{thm:general-lb-2}.}
Applying the same argument as that of Lemma~\ref{lem:elemnt-query}, we derive that the probability that $\alg$ returns different outputs on $\binst$ and $\minst$ is at most
\begin{align*}
\Prob{\alg(\binst) \neq \alg(\minst)} &\leq  \sum_{t=1}^{\card{Q}} \Prob{\ans_{\binst}(q_t) \neq \ans_{\minst}(q_t)} \leq \sum_{t=1}^{\card{Q}} P_{\eltq-\setq}(e(q_t), S(q_t)) \leq {64p_0^k \over m(1-p_0)^2} \card{Q},
\end{align*}
via the result of Lemma~\ref{lem:gen_p_cell_uniform}. Then, over the distribution in which we applied Yao's lemma, we have
\begin{align*}
\Prob{\alg \text{ succeeds}} \leq 1-{1\over 2}\pr_{\minst\sim\minstD(\binst)}[\alg(\binst) = \alg(\minst)] &\leq 1- {1 \over 2} \left(1- {{64 p_0^k \over m(1-p_0)^2}\card{Q}} \right) \\ 
	&= {1\over 2} + {32 p_0^k \over m(1-p_0)^2}\card{Q} \\
	&\leq \frac{1}{2} + \frac{32}{m}\left(\frac{8(k\apxF+2)\log m}{n}\right)^{1\over 2\apxF}\card{Q}
\end{align*}
where the last inequality follows from Lemma \ref{lem:gen-k-property}. Thus, if the number of queries made by $\alg$ is less than $\frac{m}{192}(\frac{n}{8(k\apxF+2)\log m})^{1/(2\apxF)}$, then the probability that $\alg$ returns the correct answer over the input distribution is less than $2/3$ and the proof is complete.

\end{document}